\documentclass[sn-basic,Numbered]{sn-jnl}

\usepackage{graphicx}
\usepackage{color}
\usepackage{soul}

\usepackage{amssymb,amsmath,amsthm,ragged2e}

\usepackage{algorithm}
\usepackage[]{algpseudocode}
\usepackage{tikz}
\usepackage{pgfplots}
\usepackage{verbatim}
\usepackage{mathrsfs}
\usepackage{lineno}
\usepackage[utf8]{inputenc}
\usetikzlibrary{arrows}

\usepackage{enumerate}

\usepackage{cleveref}

\usepackage{anyfontsize}

\graphicspath{{figs/}}

\pgfplotsset{compat=1.15}

\usepackage{xcolor}
\usepackage{soul}

\newtheorem{theorem}{Theorem}
\newtheorem{lemma}[theorem]{Lemma}
\newtheorem{observation}{Observation}

\newtheorem{definition}[theorem]{Definition}
\newtheorem{claim}{Claim}[theorem]
\renewcommand{\theclaim}{\arabic{claim}}
\newenvironment{claimproof}[1][\proofname\ of Claim \theclaim]{%
  \proof[#1]%
  
}{\endproof}

\Crefname{table}{Table}{Tables}
\Crefname{obs}{Observation}{Observations}
\Crefname{cor}{Corollary}{Corollaries}
\Crefname{algorithm}{Algorithm}{Algorithms}
\Crefname{lemma}{Lemma}{Lemmas}

\algnewcommand\algorithmicforeach{\textbf{for each}}
\algdef{S}[FOR]{ForEach}[1]{\algorithmicforeach\ #1\ \algorithmicdo}

\newcommand{\floor}[1]{\left\lfloor{#1}\right\rfloor}

\newcommand{\REM}[1]{}

\newcommand{\sset}[1]{\{ #1 \}}

\newcommand{\dashedrightarrow}{\dashrightarrow}
\newcommand{\abs}[1]{| #1 |}

\newcommand{\Opt}{\sf{Alg}_{\mathrm{DS}}}

\newcommand{\Gmrg}{\sf{Alg}_{\mathrm{GraphMerge}}}

\newcommand{\SEF}{\emph{Smallest Element First}}

\nolinenumbers

\begin{document}

\title{Improving Order with Queues}

\author[1,2]{\fnm{Andreas} \sur{Karrenbauer}}\email{andreas.karrenbauer@mpi-inf.mpg.de}
\author[1,2]{\fnm{Kurt} \sur{Mehlhorn}}\email{mehlhorn@mpi-inf.mpg.de}
\author[3]{\fnm{Pranabendu} \sur{Misra}}\email{pranabendu@cmi.ac.in}
\author*[1,2]{\fnm{Paolo Luigi} \sur{Rinaldi}}\email{prinaldi@mpi-inf.mpg.de}
\author[4]{\fnm{Anna} \sur{Twelsiek}}\email{anna@twelsiek.com}
\author[5]{\fnm{Alireza} \sur{Haqi}}\email{ahaqi@stanford.edu}
\author[6]{\fnm{Siavash} \sur{Rahimi Shateranloo}}\email{srahimis@caltech.edu}

\affil[1]{\orgdiv{Max Planck Institute for Informatics}, \city{Saarbrücken}, \country{Germany}}
\affil[2]{\orgdiv{Saarland University}, \city{Saarbrücken}, \country{Germany}}
\affil[3]{\orgdiv{Chennai Mathematical Institute (CMI)}, \city{Chennai}, \country{India}}
\affil[4]{\orgdiv{DAKO GmbH}, \city{Jena}, \country{Germany}}
\affil[5]{\orgdiv{Stanford University}, \city{Stanford}, \country{United States}}
\affil[6]{\orgdiv{California Institute of Technology}, \city{Pasadena}, \country{United States}}

\date{}

    
\abstract{
Given a sequence of $n$ numbers and $k$ parallel First-in-First-Out (FIFO) queues, how close can one bring the sequence to sorted order? It is known that $k$ queues suffice to sort the sequence if the \emph{Longest Decreasing Subsequence (LDS)}
of the input sequence is at most $k$. But, what if the number of queues is too small for sorting completely.  
\begin{enumerate}
	\item We give a simple algorithm, based on Patience Sort, that reduces the LDS by $k - 1$. We also show, that the algorithm is optimal, i.e., for any $L > 0$ there exists a sequence of LDS $L$ such that the LDS cannot be reduced below $L - k + 1$ with $k$ queues.
    
	\item Merging two sorted queues is at the core of Merge Sort. In contrast, two sequences of LDS two cannot always be merged into a sequence of LDS two. We characterize when it is possible and give an algorithm to decide whether it is possible. Merging into a sequence of LDS three is always possible. 
	
    \item A \emph{down-step} in a sequence is an item immediately followed by a smaller item. We give an optimal algorithm for reducing the number of down-steps. The algorithm is online. 
\end{enumerate}
    Our research was inspired by an application in car manufacturing.

\bigskip

}  %

\maketitle

\bmhead{Acknowledgements}
The research reported by AK and PLR in this paper was carried out in the context of project MoDigPro, which was supported by the European Regional Development Fund (ERDF). All authors were at the Max Planck Institute for Informatics when this research was performed. AK, KM, PM, and PLR are the main authors of the paper. AT contributed at the beginning of this research, and AH and SRS contributed near the end. We thank Leo Wennmann for the valuable conversations at the beginning of the project.

\section{Introduction}
A classic problem in computer science is the following: Given a sequence $\pi$ of $n$ numbers, and a collection of parallel queues, can we sort $\pi$ using these queues? Throughout this paper, sorting means \emph{sorting into ascending order}. 
Queues operate on a First-In-First-Out(FIFO) principle.
The elements of $\pi$ are, one by one, pushed into one of the queues, and the output sequence is generated by popping elements from one of the queues, one by one. 

This problem was studied by Tarjan~\cite{tarjan1972sorting}, inspired by a problem posed by Knuth~\cite{knuth68}, and also by Even and Itai~\cite{even1971queues}.
It was studied even earlier by
Schensted~\cite{schensted1961longest}, and later by Mallows~\cite{mallows1963patience} in the context of certain card games (Patience and Floyd's game). It is known that one can sort $\pi$ if and only if the number of queues is at least the length of the longest decreasing subsequence (LDS) of $\pi$~\cite{tarjan1972sorting}, and that Patience Sort~\cite{aldous1999longest} uses the minimum number of queues. 
The problem is related to the chromatic number of the permutation graph~\cite{golumbic2004algorithmic,even1971queues} of the sequence: the LDS of $\pi$ is equal to the chromatic number of a certain conflict graph of $\pi$, which is a permutation graph. It is also related to certain properties of the Young Tableaux~\cite{schensted1961longest,aldous1999longest}.

Patience Sort~\cite{aldous1999longest}, a sorting algorithm inspired by the card game Patience, sorts a sequence of numbers using a minimal number of queues. The input is a sequence $\pi_1$ to $\pi_n$ of $n$ distinct numbers, and the goal is to sort the numbers into increasing order. One starts with a set of empty queues $Q_1$, $Q_2$, \ldots\ and then enqueues the numbers one by one. When $\pi_i$ is to be enqueued, one determines the smallest index $j$ such that either $\pi_i$ is larger than the last element in the $j$-th queue, or the $j$-th queue is empty and then appends $\pi_i$ to the $j$-th queue. In this way, an increasing sequence is constructed in each queue. %
Once all elements are enqueued, the sorted sequence is extracted by repeatedly dequeuing the smallest first element of any queue.
We call this dequeuing strategy \SEF{} for later reference.
The number of queues required by Patience Sort is the length of the longest decreasing subsequence in the input sequence, and no algorithm can sort with fewer queues. Moreover, a permutation can be sorted again using $k$ queues if it avoids the pattern $(k+1, k,\ldots, 1)$.

In this paper we ask: What can one do if the number of queues is \emph{fixed} and maybe smaller than the LDS of the input? This is a natural generalization and an interesting combinatorial problem in its own right.
Surprisingly, not much was known about it. We also ask: How does one measure the \emph{disorder} of a sequence, and how much of it can one remove with $k$ queues?
We will consider two measures of disorder:
\begin{itemize}
    \item The length of the longest decreasing subsequence of a sequence, and 
    \item The number of \emph{down-steps} in a sequence, i.e., the number of elements that are immediately followed by a smaller element.
\end{itemize}

Let $k$ be the number of available queues, let $n$ be the length of the input sequence, and let $L$ be the LDS of this sequence. The LDS is $1$ if and only if the number of down-steps is $0$, which happens only for sorted sequences. Generally,  the number of down-steps is always at least as large as the LDS minus one; this holds since any two elements in a longest decreasing subsequence must be separated by a down-step. On the other hand, there exist sequences of length $n$, where the number of down-steps is $n/2$ while the LDS is only $2$, e.g., $2, 1, 4, 3, \ldots, n, n-1$. 

For LDS, we first give a simple algorithm based on Patience Sort that reduces the LDS from $L$ to $L - k + 1$. It can be considered as a truncated Patience Sort and runs in time $O(n \log k)$. 

\begin{theorem}[Refined Patience Sort]\label{Refined Patience Sort} Let $L \ge 1$ be arbitrary. With $k$ queues and in time $O(n \log k)$, the LDS of any sequence with length $n$ and LDS $L$ can be reduced to $\max(L - k + 1,1)$.
\end{theorem}

One cannot do better in the worst case, as the next theorem says. 

\begin{theorem}[Optimality of Refined Patience Sort]\label{intro:k-queue-thm}
 For every $L$, there exists a sequence of LDS $L$ that cannot be improved to an LDS strictly less than $L - k + 1$ with $k$ parallel queues. The sequence has length $\Theta(L^{2k - 1})$.
\end{theorem}

Theorem~\ref{intro:k-queue-thm} does not exclude the existence of an algorithm that reduces for every input $\pi$ the LDS in a way optimal for this instance. 
The LDS can certainly not be reduced below $L/k$, as at least one of the queues must contain at least that many elements from an LDS of the input, and elements in the same queue cannot overtake each other. This lower bound suggests a natural strategy: Distribute the elements of the input sequence to the $k$ queues, such that the LDS of each queue is at most $\lceil L/k \rceil$, and then merge these sequences optimally.

How does one merge non-sorted sequences optimally? We make a first step and show how to optimally merge two sequences of LDS $2$; we leave open the general question. 

\begin{theorem}\label{intro:mergeability}
	Two sequences of LDS $2$ can always be merged into a sequence of LDS 3. There exist two sequences of LDS $2$ that cannot be merged into a sequence of LDS $2$.
	 
	Further, there is an algorithm that, given two sequences of LDS $2$, either merges them into a sequence of LDS $2$, or demonstrates that this is impossible. It can be implemented to run in linear time in $n$, the sum of the length of the two input sequences.
\end{theorem}

Polynomial running time can be achieved via a reduction to the {\sc $2$-SAT} problem. Linear running time requires additional work. 

In Section~\ref{DownSteps} we turn to a second measure of disorder, namely the number of down-steps. Recall that this is the number of items in the sequence that are immediately followed by a smaller item. Intuitively, each down-step represents a `hiccup' in the sequence, and our objective is to minimize these.

We give an instance-optimal algorithm $\Opt$ for decreasing the number of down-steps, i.e., no algorithm can output a sequence with fewer down-steps than $\Opt$ does. The algorithm is \emph{online}, i.e., the decision where to enqueue a particular element $\pi_i$ does not depend on the elements after $\pi_i$ in the input sequence. Further, it is essentially the unique optimal online algorithm in the sense that if any other online algorithm $\sf A$ deviates from $\Opt$ at any step, then there is a continuation of the sequence for which $\sf A$ produces an output with more down-steps than $\Opt$. The algorithm runs in time $O(n \log k)$, where $n$ is the length of the input sequence, and $k$ is the number of available queues. We thus completely characterize the power of $k$ parallel queues in reducing the number of down-steps. 

\begin{theorem}\label{intro:downsteps}
	 There is an algorithm $\Opt$ that constructs an output sequence with a minimum number of down-steps. It is online and can be implemented to run in time $O(n \log k)$. Further, if any other online algorithm $\sf A$ deviates from $\Opt$, the input sequence can be extended such that $\sf A$ does strictly worse than $\Opt$. 
\end{theorem}

\subsection*{Our Motivation and Related Results}\label{Motivation and Related Work}
Our research
was motivated by an application in car manufacturing~\cite{karrenbauer2025optimizingcarresequencingmixedmodel}.
The production in a car factory consists of three main stages: body (welding of the body), paint (painting the body), and assembly (adding engines, seats, \ldots). Nowadays, cars can be configured in many different variants (different colors, different engines, different interiors, different entertainment systems, different sensors, \ldots), and the sequence of cars passing through the production process must satisfy the constraints of the factory floor and machinery. For example, no two subsequent cars may have the largest engine, only one out of three subsequent cars may have a roof-top, color changes cost money, and so on. Planning a sequence violating a (near) minimum number of constraints is an active area of research~\cite{GagneGravel,MorinGagne,Twelsiek}.

Even if a sequence that obeys all constraints is available at the beginning of the production process, this sequence will, in general, be perturbed in body and paint, e.g., because a welding point is not perfect and needs to be redone, or the paint is not perfect and needs to be retouched. For this reason, there are buffers between body and paint and paint and assembly, where the perturbed sequence can be improved to satisfy more constraints. In some factories, these buffers consist of a set of parallel queues that operate according to the First-In-First-Out (FIFO) principle. This leads to our research question: ``How much disorder can one remove using a given set of queues?'' Surprisingly, we could not find any literature that addresses this problem for the two measures of order we consider, i.e., LDS and number of down-steps. 

Boysen and Emde~\cite{Boysen:stack-blockage} study a closely related problem under the name \emph{parallel stack loading problem (PSLP)}. They consider only the enqueuing phase, and the goal is to minimize the total number of down-steps in all queues.
They show that miminizing the number of down-steps in the enqueued sequences is NP-complete if the queue lengths are bounded; \cite{Boge-Knust} gives an alternative NP-completeness proof.
Gopalan et al.~\cite{Gopalan} study a streaming algorithm with the goal of approximately miminizing the number of edit-operations required for transforming an input sequence into a sorted sequence. This measure is also called the \emph{Ulam metric}. For a sequence of length $n$, this measure is $n$ minus the length of the longest decreasing subsequence. Parallel algorithms for computing the edit and the Ulam distance are given by Boroujeni et al.~\cite{Boroujeni}.
Aldous and Diaconis\cite{aldous1999longest} discuss the many interesting properties of Patience Sort. Chandramouli et al.~\cite{Chandramouli} give a high-performance implemention of Patience Sort for almost sorted data.

\paragraph*{Organization of the Paper}
In Section~\ref{LDS}, we prove Theorems~\ref{Refined Patience Sort} and~\ref{intro:k-queue-thm}, in Section~\ref{Mergability}, we prove Theorem~\ref{intro:mergeability}, and finally in Section~\ref{DownSteps}, we give a proof of Theorem~\ref{intro:downsteps}. We conclude in Section~\ref{conclusion} with some open problems.

\section{Reducing the Longest Decreasing Subsequence}\label{LDS}

We begin with the first measure of disorder of a sequence, namely the length of the longest decreasing subsequence (LDS).
We are interested in the following question: How much can one reduce the LDS of a sequence using $k$ parallel queues?
We first present a simple algorithm based on Patience Sort~\cite{aldous1999longest} that reduces the LDS from $L$ to $L - k + 1$. We then show that this is best possible in the worst case, by exhibiting, in Section~\ref{Lower Bounds}, for each integer $L$, a sequence of LDS $L$ whose LDS cannot be decreased below $L - k + 1$ with $k$ queues.

\subsection{Reducing LDS via Patience Sort}\label{Alternative Queuing}

We give a proof of Theorem~\ref{Refined Patience Sort}. Our algorithm is based on Patience Sort. Patience Sort sorts a sequence using LDS-many queues. In each queue, a decreasing sequence is constructed, and these are then merged. We use the first $k - 1$ queues as in Patience Sort, i.e., construct a decreasing sequence in each one of them, and use the $k$-th queue for all the remaining queues that Patience Sort would open.

\begin{claimproof}[Proof of Theorem~\ref{Refined Patience Sort}]
    Let us name the queues $Q_1, Q_2, \ldots, Q_k$. 
    The elements of the input sequence are processed one by one from the first to the last one. If $Q_k$ is still empty or the current element is larger than the last one in $Q_k$, put the current element in the same queue as Patience Sort would do, i.e., choose the smallest $i$ such that $Q_i$ is empty or the current element is larger than the last element of $Q_i$. Otherwise, put the current element in $Q_k$.
	This algorithm runs in $O(n \log k)$ time by using binary search because the last elements of $Q_1,\ldots, Q_k$ always form a decreasing sequence.
    Moreover, the first $k-1$ queues have LDS $1$.
    This leaves an LDS of $L-k+1$ for the $k$-th queue because Patience Sort partitions the input sequence into $L$ increasing subsequences, say $Q'_1,\ldots,Q'_L$, where $Q_i = Q'_i$ for $i \in [k-1]$.
    Let us have a closer look at $Q_k$. Let $x$ and $y$ be consecutive elements in $Q'_k$. Then, all elements in $Q_k$ between $x$ and $y$ are smaller than $x$, as an element larger than $x$ would have been added to one of the queues $Q'_1$ to $Q'_k$. 

    We now merge the queues $Q_1$ to $Q_k$ using the smallest element first strategy. We claim this merges $Q'_1,\ldots,Q'_{k}$ into an increasing sequence $S$, thus yielding an output that can be partitioned into $L-k+1$ increasing sequences, namely $S$ and $Q'_{k+1}$ to $Q'_L$. Indeed, assume we have just output an element $x$ in $Q'_k$ from $Q_k$, and let $y$ be the next element from $Q'_k$ in $Q_k$. Then the first elements of $Q_1$ to $Q_{k-1}$ are larger than $x$, and all elements in $Q_k$ before $y$ are smaller than $x$. Thus, they will be output next until $y$ is at the front of $Q_k$. Now the merge of $Q_1$ to $Q_k$ will continue in the same way as a merge of $Q'_1$ to $Q'_k$ until $y$ is dequeued.
    To implement the dequeuing in $O(n \log k)$ time, we maintain a binary heap containing the indices from $1$ to $k$ such that the heap property is considered w.r.t.\ the first elements of the respective queues (where the first element of an empty queue is defined to be $\infty$). Hence, the queue with the smallest first element can be determined in constant time. After dequeuing the smallest element, the heap property can be restored in $O(\log k)$ time, which yields $O(n \log k)$ time for dequeuing all elements.

\end{claimproof}

The above algorithm is simple. It is natural to ask whether a more sophisticated algorithm can give better results. The answer is no in the worst case. We show in Section~\ref{Lower Bounds} that the algorithm above is instance-optimal and no algorithm can do better on all inputs. Our lower bound does not exclude that one obtains a smaller LDS for some/most inputs, perhaps by better balancing the LDS of the queues, because the LDS in the output cannot be smaller than the maximum LDS over the $k$ queues. For example, this could be done by putting for each $i \in [k]$ into $Q_i$ all elements from $Q'_j$ for each $j \in [L]$ with $j \equiv i \bmod k$. Then, the LDS of each of $Q_1,\ldots, Q_k$ is at most $\lceil L/k \rceil$.

We would then also have to look at a different merging strategy, as the smallest element first strategy does not always yield the optimum merge, as the following example shows. Consider $k=2$ and the input sequence $5, 4, 1, 3, 2$ with LDS $L=4$. The \emph{modulo} procedure described above yields the two queues
\begin{align*}
 Q_1 &: 5,1,3 & Q_2 : 4, 2, 
\end{align*}
and the smallest element first strategy merges them to $4, 2, 5, 1, 3$ with LDS $3$. However, we can also obtain $4, 5, 1, 2, 3$ with LDS $2$, which is optimal, as both queues have LDS $2$, and at least one of the two queues must have an LDS of at least $L/k$.

Using two queues, two sequences of LDS two can always be merged into a sequence of LDS three (Theorem~\ref{intro:mergeability}) and maybe into a sequence of LDS two (the preceding example). In Section~\ref{Mergability} we show that two cannot always be obtained and derive an optimal merging strategy.

\subsection{The Limits of Reducing LDS}\label{Lower Bounds}

We show that for any $L$ and $k$, there is a sequence with LDS $L$ such that it is impossible to reduce the LDS  to a value lower than $L - k + 1$ with $k$ queues. The length of the sequence is $\Theta(L^{2k-1})$. We first present a construction for $k=2$ (Section~\ref{lower2}) and then extend it to arbitrary $k$ (Section~\ref{lowerk}). We then show the bound on the length of the constructed sequence.
We begin with the following definition.
\begin{definition}
    For a sequence $A$, let $\Lambda_k(A)$ be the lowest LDS of any sequence that can be obtained by enqueuing and dequeuing $A$ using $k$ queues.
\end{definition}

 In our construction, we use the direct and the skew sum of permutations. Given two permutations $\pi$ and $\sigma$ and their respective permutation matrices $M_{\pi}$ and $M_{\sigma}$, the \emph{direct sum} and the \emph{skew sum} of these two permutations, in terms of permutation matrices, take the form
\begin{align*}
M_{\pi} \oplus M_{\sigma} &= \begin{bmatrix}
M_{\pi} & 0 \\
0 & M_{\sigma}
\end{bmatrix} \qquad \text{and}&
M_{\pi} \ominus M_{\sigma} &= \begin{bmatrix}
0 & M_{\pi} \\
M_{\sigma} & 0
\end{bmatrix}\text{, respectively.}
\end{align*}
Equivalently, we can define the \emph{direct sum} of two permutations $\pi$ and $\sigma$ as
$$
\pi \oplus \sigma = \pi \mathbin\Vert \sigma',
$$
where $\sigma'$ is obtained from $\sigma$ by adding $\abs{\pi}$ to all its elements, and the \emph{skew sum} as
$$
\pi \ominus \sigma   = \pi' \mathbin\Vert \sigma,
$$
where $\pi'$ is obtained from $\pi$ by adding $\abs{\sigma}$ to all its elements.

\subsubsection{Lower Bound for two Queues} \label{lower2}

We give a proof of Theorem~\ref{intro:k-queue-thm} for the special case of two queues. We exhibit for each $L \in \mathbb{N}$ a permutation $\pi$ with LDS $L$
such that any manner of enqueueing it into two queues and then dequeueing it yields a permutation of LDS at least $L - 1$.
We prove this claim by induction on $L$. The following observation is simple but crucial for the lower bound construction.

\begin{observation}\label{skewdirect}
	\begin{itemize}
		\item Consider the direct sum $\pi \oplus \sigma$. In the resulting sequence, all the elements that originally belonged to $\pi$ are lower than those that originally belonged to $\sigma$. Furthermore, the LDS of $\pi \oplus \sigma$ is the maximum of the LDS of $\pi$ and the LDS of $\sigma$.

        \item Consider the skew sum $\pi \ominus \sigma$. In the resulting sequence, all the elements that originally belonged to $\pi$ are greater than those that originally belonged to $\sigma$. Furthermore, the LDS of $\pi \ominus \sigma$ is the sum of the LDS of $\pi$ and the LDS of $\sigma$.
	\end{itemize}
\end{observation}

\begin{claimproof}[Proof of Theorem~\ref{intro:k-queue-thm} for $k = 2$] For $L \le 2$, the claim is obvious. For $L = 3$, consider $\pi = [3, 2, 1]$. The LDS is three and with only two queues it is not possible to obtain an output LDS lower than two, since two of the elements have to be enqueued into the same queue. 
For $L \ge 3$, we define $A_{L + 1}$ as
	\[ A_{L + 1} = (A_L \oplus \underbrace{D_L \oplus \dots \oplus D_L}_{L - 1 \text{ times}}) \ominus [1],\]
	where $D_L$ denotes the decreasing sequence of length $L$, and the corresponding permutation matrix is simply the anti-diagonal matrix of size $L$.
	The permutation matrix of $A_{L+1}$ is of the form
	$$
	A_{L + 1} = \begin{bmatrix} 
		0 & A_{L} & 0 & 0 & \dots & 0 \\
		0 & 0 & D_{L} & 0 & \dots & 0 \\
		0 & 0 & 0 & D_{L} & \dots & 0 \\
		\vdots & \vdots & \vdots & \vdots & \ddots & \vdots \\
		0 & 0 & 0 & 0 & \dots  & D_{L} \\
		1 & 0 & 0 & 0 & \dots  & 0 
	\end{bmatrix}.
	$$
	According to Observation~\ref{skewdirect}, $A_{L + 1}$ has LDS $L+1$. We prove $\Lambda_2(A_{L + 1}) \geq L$. We refer to $B_i$ and $C_i$ as the subsequences of the $i$-th copy of $D_L$ that are enqueued into the first and the second queue, respectively, where $1 \leq i \leq L - 1$. In the same way, we call $X$ and $Y$ the sequences in the first and second queue enqueued from $A_L$.
	W.l.o.g.~assume that the element $1$ is inserted into the first queue. So after enqueuing, we have 
    \[   Q_1: X, B_1, \ldots, B_{L-1}, 1 \qquad  Q_2: Y, C_1, \ldots, C_{L-1}, \]
    and 1 is dequeued after $X$ and all the sequences $B_i$.
    We consider two cases:
	\begin{itemize}
		\item The element $1$ is dequeued after the last element of $Y$: Since the elements of $X$ are in the same queue as the element $1$, the $1$ is also dequeued after the last element of $X$. Therefore, all the elements of $X$ and $Y$ will appear before the element $1$ in the final sequence. By the inductive hypothesis, merging the two sequences $X$ and $Y$ will create a decreasing sequence of length $L - 1$. Since the element $1$ is dequeued afterwards, we obtain a decreasing sequence of length $L$.
		\item The element $1$ is dequeued before the last element of $Y$:
		Denote the LDS of the sequences $B_1, B_2, \dots, B_{L - 1}$ as $\ell_1, \dots , \ell_{L-1}$, respectively. If any $\ell_i \geq L - 1$, the sequence $B_i \cup \{1\}$ ($B_i$ followed by 1) will have an LDS greater or equal to $L$, satisfying the claim of the theorem. Similarly, if any $\ell_i = 0$, we have $C_i = D_L$, and the LDS is at least $L$. We can, therefore, assume that $1 \leq \ell_i \leq L - 2$ for all $i \in [L-1]$. By the pigeonhole principle, there exists a pair $i, j$ such that $i < j$ and $\ell_i = \ell_j$. Notice that, by the construction of $A_{L + 1}$, all the elements of $B_j$ are greater than all the elements of $C_i$. The element $1$ is dequeued after all elements of $B_j$ and before the last element of $Y$, and, hence, it is dequeued before $C_i$. Since $\ell_i=\ell_j$, the LDS of the concatenation of $B_j$ and $C_i$ is equal to $L$, which proves the claim.
	\end{itemize}
\end{claimproof}

How large is $A_L$ as a function of $L$?  It is composed of $A_{L - 1}$, $L - 2$ sequences $D_{L-1}$ of length $L - 1$, and finally the element $1$. Therefore,  $size(A_{L}) = size(A_{L - 1}) + (L - 2) \cdot (L - 1) + 1$. Together with $size(A_{3}) = 3$, we obtain $size(A_L) = \frac{L(L^2 - 3L + 5)}{3} - 2 =\Theta(L^3)$.

\subsubsection{Lower Bound for \texorpdfstring{$k$}{k} Queues} \label{lowerk}
In this part, we present a permutation $\pi$ with  LDS $L$ such that $\Lambda_k(\pi) \geq L - k + 1$. We first describe the sequence, then we prove some preparatory lemma, and finally, we establish the theorem. 

\begin{definition}
	For $k \ge 1$ and $L \ge 1$, define permutation $B_L^k$ by the following rules:
	\begin{itemize}
		\item For any $L$, set $B_L^1$ to be a decreasing sequence of length $L$.
		\item For each $1 \leq L \leq k + 1$, set $B_L^k$ to be a decreasing sequence of length $L$.
		\item For every $L \geq k + 1$, and $k \geq 2$, let
		$$
		B_{L + 1}^k = (B_L^k \oplus \underbrace{B_L^{k - 1} \oplus \dots \oplus B_L^{k - 1}}_{L - 1 \text{ times}}) \ominus [1].
		$$
	\end{itemize}
\end{definition}

Observe that, for $k=2$, the definition above is the same as in the previous subsection. Similarly to the case for $k=2$, one can easily show that the LDS of $B_L^k$ is exactly $L$. So, it remains to be proven that $\Lambda_k(B_L^k) \geq L - k + 1$ for all $L \geq k + 1$. Note that for $L \leq k$, it is possible to yield a sorted sequence as the result. It is also easy to see that for $L = k+1$, the lowest obtainable LDS is $2 = L - k + 1$, because, by the pigeonhole principle, at least two elements will be enqueued on the same queue.

Now we state and prove some lemmas that will be useful later.

\begin{definition}
	Sequences $A'$ and $A''$ are called a \emph{decomposition} of $A$ if $A'$ and $A''$ are disjoint subsequences of $A$ and $A$ is an interleaving of $A'$ and $A''$.
\end{definition}

\begin{lemma}
	\label{two_seqs}
	For any decomposition $A$, $B$ of some sequence $C$, if $\Lambda_s(A) = \ell$ and $\Lambda_{s'}(B) = t$, then $\Lambda_{s + s' - 1}(C) \leq \ell + t$.
\end{lemma}

Before proceeding with the proof, we make some remarks and claims so that the proof becomes more intuitive and easier to verify. Consider a sequence $D$ being completely enqueued and completely dequeued using $k$ queues and resulting in a sequence $D'$. Let us denote the enqueuing protocol, that is, the policy used to enqueue $D$, $S_e$ and the dequeuing protocol $S_d$.
We can encode $S_e$ into a sequence, where the $i$-th element is the index of the queue into which the $i$-th element of $D$ is enqueued. For example, the string $(1, 3, 2)$ encodes a protocol that enqueues the first element into $Q_1$,
the second element into $Q_3$, 
and the third element into $Q_2$. 
We can encode $S_d$ in the same way, but this time, each element indicates the index of the queue from which $i$-th element of the output is dequeued.

\begin{claim}
	\label{same_prot}
	The following two ways to execute the protocols are equivalent and result in the same sequence $D'$.
	\begin{itemize}
		\item First, enqueue all the elements of $D$ according to $S_e$, and, afterwards, dequeue them according to $S_d$.
		\item If the next element in $S_d$ is available in the respective queue, dequeue it; otherwise, enqueue the next element according to $S_e$. Repeat until the starting sequence and all the queues are empty.
	\end{itemize}
\end{claim}
\begin{claimproof}
	The first case coincides with the construction of $S_e$ and $S_d$. Therefore, it trivially leads to $D'$.
	
	Assume the resulting sequence $D''$ of the second case is different and that the first position where it differs is $i$. It means that the $i$-th element that gets dequeued is different in $D'$ and $D''$. Since we use the same protocol $S_d$ in both cases, the queue from which this element is dequeued is the same. Since we also execute the same enqueuing protocol $S_e$ as well, the relative ordering of the elements in any queue is preserved and cannot result in different elements at the same position.
\end{claimproof}
\begin{claim}
	\label{one_empty}
	In the second case of Claim~\ref{same_prot}, at each step, there is always at least one queue that is either empty or contains a single element that is, moreover, dequeued in the next step.
\end{claim}
\begin{claimproof}
 If after a time step $t$ all queues are non-empty, the element that was enqueued at time $t$ will be dequeued at time $t+1$.
\end{claimproof}

Notice that the queue in Claim~\ref{one_empty} may change over time. However, the protocol can be modified to achieve this. We follow the terminology of \cite{expressqueue1973} and denote this queue as \emph{express queue}. So, the express queue has the property that it never contains more than one element, and an element enqueued into it is immediately dequeued in the next time step. 

\begin{claim}
	\label{express_queue}
	One can modify the protocols $S_e$ and $S_d$ such that the result is the same and such that there is a fixed queue that is always either empty or contains a single element that is immediately dequeued. 
\end{claim}
	\begin{claimproof}
		The only way the queue of Claim~\ref{one_empty} can change is that there is more than one empty queue that blocks the dequeuing protocol from being executed; for example, assume $Q_1$ is currently the queue of Claim~\ref{one_empty}. Then, $Q_2$ gets emptied, and the dequeuing protocol requires an element from $Q_2$. If the enqueuing protocol enqueues elements into $Q_1$, the dequeuing protocol is still blocked by $Q_2$, which becomes the new queue of Claim~\ref{one_empty}. In order to fix an {express queue}, the enqueuing and dequeuing protocol must be modified such that whenever there is more than one empty queue, the last queue to receive an element from the enqueuing protocol is the desired {express queue}. By doing so, the fixed queue cannot be blocked by other empty queues, because it is the last one to stay empty.
	\end{claimproof}
We proceed with the proof of Lemma~\ref{two_seqs}.
\begin{proof}
	After enqueuing and dequeuing $A$ using $s$ queues, we obtain an output sequence $A'$ with a LDS of $\ell$. Denote the protocol to get this result as $S'$. Similarly, denote the protocol for $B$ that results in $B'$ with a LDS of $t$, as $S''$. Now modify $S'$ and $S''$ according to Claim~\ref{express_queue} such that they both have an {express queue}. W.l.o.g. we can assume that the {express queue} is the first queue in both cases. Now, we show how to enqueue-dequeue $C$ using $s + s' - 1$ queues such that the resulting sequence has a LDS of $\ell + t$.
    
    Enqueue the $i$-th element $e$ of $C$ as follows:
	\begin{itemize}
		\item If $e$ is an element of $A$, enqueue it according to $S'_e$.
		\item If $e$ is an element of $B$, let $q \in [s']$ be the index of the queue, where it is supposed to be enqueued according to $S''_e$:
		\begin{itemize}
			\item If $q = 1$, i.e., the {express queue}, enqueue $e$ in the first queue.
            \item Otherwise, enqueue $e$ in the queue with index $q + s - 1$, and modify the dequeuing protocol $S''_d$ accordingly.
		\end{itemize}
	\end{itemize}
	We can now dequeue the elements of $C$ according to $S'_d$ and $S''_d$, where the origin of the head of the first queue determines which of the two is currently active. We thereby obtain a sequence $C'$ with an LDS that is at most the sum of the LDS of $A'$ and $B'$, i.e., $\ell + t$.
\end{proof}

\begin{lemma}
	\label{ldssubseqs_4plus}
	For any decomposition $A'$, $A''$ of some sequence $A$, if $\Lambda_k(A) = L$ and $\Lambda_{s}(A') = \ell$, with $s < k$, then $\Lambda_{k - s + 1}(A'') \geq L - \ell$.
\end{lemma}

\begin{proof}For sake of contradiction, assume $t := \Lambda_{k - s + 1}(A'') < L - \ell$. By Lemma~\ref{two_seqs}, we can enqueue-dequeue $A$ with $k = s  +  k - s + 1 - 1$ queues such that the LDS of the result is at most $\ell + t < \ell + L - \ell = L$, a contradiction to $\Lambda_k(A) = L$. Therefore, $t \geq L - \ell$.
\end{proof}

We will apply the preceding Lemma to the sequence $B_L^{k-1}$. We know, by induction hypothesis, that with $k - 1$ queues, we cannot enqueue-dequeue it into a sequence with LDS lower than $L - k + 2$. Formally, $\Lambda_{k - 1}(B_L^{k-1}) = L - k + 2$. The Lemma gives us information on what happens if we have a total of $k$ queues. If we decompose $B_L^{k-1}$ into two sequences, and we enqueue-dequeue the first one using $s$ queues and the second one using $k-s$ queues, the Lemma tells us that the sum of the resulting LDSs will still be at least $L - k + 2$.

We now come to the proof of Theorem~\ref{intro:k-queue-thm}. We show that for 
$k \geq 1$ and $L \geq k + 1$, there exists a sequence $A$ of LDS $L$ such that $\Lambda_k(A) \geq L-k+1$. The length of $A$ is $\Theta(L^{2k-1})$.

\newcommand{\cancelblue}[1]{\textcolor{blue}{\sout{#1}}}

The proof is similar to the case $k=2$ presented above. Just like for $k=2$, we identify two possible scenarios on when the element $1$ is dequeued. Finally, we show that in both cases the claim holds. On a conceptual level, this more general proof strictly follows the simpler case.
\begin{claimproof}[Proof of Theorem~\ref{intro:k-queue-thm} for general $k$]
	We use double induction on $L$ and $k$.
	First, we set the base case for $k = 1$, and we prove the claim for every $L$. Recall that, in this case, $B_L^1$ is a decreasing sequence of length $L$. Obviously, with a single queue, the sequence does not change, and the resulting sequence has an LDS of $L = L - k + 1$.
 
	Now, we set the base case for $L = k + 1$ and prove the claim for every $k$. We have already argued that, by the pigeonhole principle, the lowest possible obtainable LDS for $B_{k + 1}^k$ with $k$ queues is $2 = L - k + 1$.
 
	Now, we proceed with the induction step and assume $k \geq 2$ and $L \geq k + 1$. We want to prove $\Lambda_k(B_{L + 1}^k) \geq L - k + 2$. By induction hypothesis:
	\begin{itemize}
		\item $\Lambda_k(B_L^k) \ge L - k + 1$.
        \item $\Lambda_{k - 1}(B_{L}^{k - 1}) \ge L - (k - 1) + 1 = L - k + 2$.
	\end{itemize}
    Let $P_i$ denote the $i^{th}$ copy of $B_L^{k - 1}$ and $P_i^j$ the elements of $P_i$ that are enqueued to the $j^{th}$ queue with $1 \leq i \leq L - 1$ and $1 \leq j \leq k$.
 Similarly, let $P_0^{j}$ denote the elements of $B_L^k$ that are enqueued into the $j$-th queue. In the proof for $k = 2$, we used $C_i$ for what is now called $P_i^2$ and $X$ for what is now called $P_0^1$.
 W.l.o.g.\ assume that the element $1$ is enqueued into the first queue.
	Before proceeding, we prove another intermediate claim.
	
	\begin{claim}
		\label{cases_4}
		At least one of the following is true:
		\begin{enumerate}
			\item All the sequences $P_0^{j}$ are dequeued before the element $1$;
            \item There is a set $I$ of at least $L - k + 1$ indices $i_1, i_2, ..., i_{L-k+1}$ such that for each $i \in I$: For all $j \in [k]$, $P_i^{j}$ is either completely dequeued before $1$ or completely dequeued after $1$. Moreover, for each $i \in I$, there exist $j', j'' \in [k]$ such that $P_i^{j'}$ is completely dequeued before $1$, and $P_i^{j''}$ is completely dequeued after $1$. 
		\end{enumerate}
	\end{claim}
	\begin{claimproof}
        Assume we are not in case 1, i.e., $1$ is dequeued before the last element of some $P_0^p$. Note that $p > 1$, since $1$ is dequeued after all elements in the first queue. Also, $P_i^1$ is dequeued before $1$ for all $i$, since $1$ is dequeued after all elements in the first queue, and all $P_i^p$ are dequeued after $1$, for $i \geq 1$.
        
        Say that a queue $j$ kills a copy $P_i$ if $1$ is dequeued after the first element of $P_i^j$ and before the last of $P_i^j$. Then, queues $1$ and $p$ kill no copy, and any other queue kills at most one copy. So $L-k+1$ copies are not killed, as we have $L - 1$ copies but only $k - 2$ queues that can kill a copy.
	\end{claimproof}
	
	We now consider the two cases of Claim~\ref{cases_4} separately:
	\begin{enumerate}
		\item Recall that the sequences $P_0^{j}$ are the result of enqueuing the elements of $B_L^k$. By the inductive hypothesis, $\Lambda_k(B_L^k) \ge L - k + 1$. By assumption, $1$ is output after all elements of $B_L^k$, and, hence, an LDS of at least $L - k + 2$ will result.

         \item By assumption, there is a set $I$ indices $i_1, i_2, ..., i_{L-k+1}$ such that $1$ is either dequeued before or after all items of each $P_i^j$, for $i \in I$. Also $1$ is dequeued after all elements in $P_i^1$, $i \in I$, and before all elements of $P_i^p$, $i \in I$, for some $p > 1$.

        Therefore, for each $i \in I$, there is a set $S_i$ such that
        \begin{enumerate}
            \item $S_i \subset [k]$ and $\emptyset \neq S_i \neq [k]$, and
            \item all $P_i^j$ with $j \in S_i$ are dequeued before $1$, and all $P_i^j$ with $j \not\in S_i$ are dequeued after $1$. Note that $1 \in S_i$, and $p \not\in S_i$.
        \end{enumerate}

		\newcommand{\ensemble}[1]{P_{#1}^{j \in [s]}}
		\newcommand{\newensemble}[1]{P_{#1}^{j \in [s_{#1}]}}
		\newcommand{\newnewensemble}[1]{P_{#1}^{S_{#1}}}
		
		From now on, we will use $i$ to refer to any index in $I$. If some $P^j_i$ is empty for some $j$, $P_i$ is enqueued into fewer than $k$ queues, and, hence, the result has an LDS of at least $L - k + 2$ by induction hypothesis.
		
		We use $\newnewensemble{i}$ to denote the collection of $P_i^j$, $j \in S_i$. Let $\ell_i$ be the minimum LDS that any dequeuing of $\newnewensemble{i}$ can achieve.\\ Since $P_i^j$ is non-empty for all $j$ and $S_i \neq \emptyset$, $\ell_i \ge 1$. 
		
		Assume next that  there is an $i$ such that $\ell_i \ge L - k +1$. Then, together with the element 1, we obtain a decreasing subsequence of length $L - k + 2$. 
		
		So, we are left with the case that $1 \le \ell_i \le L - k$ for all $i$. Therefore, there must be indices $u$ and $v$ such that $1 \le u < v \le L - k + 1$ and $\ell_u = \ell_v$. Note that all elements in $P_u$ are greater than all elements in $P_v$. Note also that we are not claiming that $P_u$ and $P_v$ are enqueued in the same way. We are only claiming that both enqueuings give the same LDS. 
		
        We next apply Lemma~\ref{ldssubseqs_4plus} to $A = P_u$, where $A'$ is the subsequence of $P_u$ that is enqueued into the queues in $S_u$ and $A''$ the subsequence that is enqueued into the other queues. Since $\Lambda_{k - 1}(B_L^{k - 1}) = L - k + 2$ and $\Lambda_{|S_u|}(A') = \ell_u$, we have $\Lambda_{k - |S_u|}(A'') = \Lambda_{k - 1 - |S_u| + 1}(A'') \ge L - k + 2 - \ell_u$. In the same way, we can apply the lemma to $B = P_v$, where $B'$ is the subsequence of $P_v$ that is enqueued into the queues in $S_v$, and $B''$ the subsequence that is enqueued into the other queues. We get that $\Lambda_{k - |S_v|}(B'') \ge L - k + 2 - \ell_v$.
		
        Now we are done. In the output, we have a decreasing subsequence of length $\ell_v$ resulting from the dequeuing of $\newnewensemble{v}$, followed by a decreasing sequence of length $L - k + 2 - \ell_u = L - k + 2 - \ell_v$ resulting from the dequeuing of the $P_u^j$, $j \not\in S_u$.
  Note that all elements in the former sequence are larger than all elements in the latter sequence. Hence, an LDS of $L - k + 2$ results.
  \end{enumerate}
 
 This completes the proof.
\end{claimproof}

\subsubsection{Growth of the Sequence}
We show that the length of the sequence $B_L^k$ is $\Theta(L^{2 k - 1})$. Denote by $S_L^k$ the length of $B_L^k$. For $k=2$, we have already seen that $S_L^2 = \frac{L(L^2 - 3L + 5)}{3} - 2 \in \Theta(L^3)$. For $k \ge 3$, we have 
\[
S_L^k = S_{L-1}^k + (L - 2) S_{L-1}^{k-1} + 1.
\]
By expanding the recursive formula on the first term, we obtain
\[
S_L^k = L + \sum_{i=k + 1}^{L-1} (i - 1) S_{i}^{k - 1},
\]
where we used $S_{k + 1}^k = k + 1$.

It is now easy to complete the proof by induction.
The base case $S_L^1$ is trivial, as $S_L^1 = L$. We now assume that $S_L^{k - 1} \in \Theta(L^{2k - 3})$. The term that will result in the highest degree is $\sum_{i=k + 1}^{L-1} (i - 1) S_{i}^{k - 1}$. In particular, $(i - 1) S_{i}^{k - 1} \in \Theta(i^{2k - 2})$, by the inductive hypothesis. Since we sum the indices $i$ up to $L$, we get a polynomial one degree higher. Hence, $S_L^k \in \Theta(L^{2k - 1})$.

\section{Mergeing Two Sequences of LDS Two}\label{Mergability}

Merging two (or more) sorted sequences is central to some efficient sorting algorithms such as Mergesort~\cite{cormen2022algorithms}.
The LDS of a sorted list is one, and it is always possible to merge two (or more) queues of LDS one into an output sequence of LDS one. We ask: How well can one merge sequences of LDS greater than one. We initiate a study of this question and give an answer for two sequences of LDS two. They can always be merged into a sequence of LDS three and sometimes into a sequence of LDS two. We give a linear time algorithm that decides whether the latter is possible. 

\subsection{Merging Two Sequences of LDS Two in Polynomial Time}

We give a proof of Theorem~\ref{intro:mergeability}, however without establishing linear running time. Linear running time will be established in the next section. We will define a precedence graph on the elements of the two sequences. There will be two kinds of edges: strict edges and tentative edges. The tentative edges come in pairs sharing an element. We show that an output sequeence has LDS at most two iff it obeys all strict edges and one edge in each pair of tentative edges. We then derive an algorithm to decide whether such a selection of tentative edges is possible.

We have two sequences $Q_1$ and $Q_2$, each of LDS at most two. We want to decide whether $Q_1$ and $Q_2$ can be merged into a sequence of LDS at most two. A \emph{down-pair} in a sequence is a pair $a,b$ with $a > b$ and $a$ preceding $b$ in the sequence. Then, $a$ and $b$ are the \emph{left element} and \emph{right element} of the down-pair, respectively. 

We will define a \emph{precedence graph} on the elements of $Q_1$ and $Q_2$. There are two kinds of edges: \emph{strict edges} and \emph{tentative edges}. A strict edge $x \rightarrow y$ indicates that $x$ must be before $y$ in the output. Tentative edges come in pairs that share a common element. If $b \dashrightarrow e \dashrightarrow a$ are a pair of tentative edges, then either $b$ must be before $e$ or $e$ must be before $a$ in the output; so the arrangement $a, e, b$ is forbidden. 

In a sequence of LDS at most two, we have three kinds of elements: left elements in a down-pair, right elements in a down-pair, and elements that are in no down-pair. We call such elements \emph{solitaires}. Note that no element can be a left and a right element, as this would imply a decreasing sequence of length three. For a solitaire, all preceding elements are smaller and all succeeding elements are larger. We have the following precedence edges.

\begin{enumerate}[a)]
	\item For each sequence and elements $x$ and $y$, we have the edge $x \rightarrow y$ if $x$ precedes $y$ in the sequence.
	\item Assume $a,b$ is a down-pair in one of the sequences and $c,d$ a down-pair in the other sequence.
	We have the edges $\min(a,c) \rightarrow \max(a,c)$ and $\min(b,d) \rightarrow \max(b,d)$.\footnote{We give two arguments for why these constraints are necessary. Assume $a > c$ and $a$ precedes $c$ in the merge. Then, $a,c,d$ is a decreasing sequence of length three. Assume $b > d$ and $b$ precedes $d$ in the merge. Then, $a, b, d$ is a decreasing sequence length three. Alternatively, two down-pairs $a,b$ and $c,d$ can either be disjoint, interleaving, or nested. Assume w.l.o.g.\ that $a > c$. If the pairs are nested, i.e., $a > c > d > b$, the only possible order is $c,a,b,d$. This is guaranteed by the edges $c \rightarrow a \rightarrow b \rightarrow d$. If the pairs are interleaving, i.e., $a > c > b > d$, or disjoint, i.e., $a > b > c > d$, the two possible orders are $c,a,d,b$ and $c,d,a,b$. This is guaranteed by the edges $a \rightarrow b$, $c \rightarrow d$, $c \rightarrow a$, and $d \rightarrow b$.} 
	\item Between two solitaires in different sequences, there are no constraints.
	\item Assume $a,b$ is a down-pair in one of the sequences and $e$ a solitaire in the other. We generate enough edges so that the elements cannot form a decreasing sequence of length three. 
	\begin{enumerate}[1)]
		\item If $e > a$, then $a \rightarrow e$.
		\item If $e < b$, then $e \rightarrow b$.
		\item If $a > e > b$, then we have the pair $e \dashrightarrow a$ and $b \dashrightarrow e$ of tentative edges. We call them \emph{companions} of each other. An edge may have many companions. For example, if $a,b$ and $a',b$ are down-pairs with $a > e > b$ and $a' > e > b$,  $e \dashrightarrow a$ and $e \dashrightarrow a'$ are companions of $b \dashrightarrow e$. 
		Similarly, if $a,b$ and $a,b'$ are down-pairs with $a > e > b$ and $a > e > b'$,  $b \dashrightarrow e$ and $b' \dashrightarrow e$ are companions of $e \dashrightarrow a$.
		
		Tentative edges only end in solitaires or left endpoints.\end{enumerate}
\end{enumerate}

\begin{lemma}\label{lemma:goodgraph} If the output sequence obeys all strict edges and one tentative edge of each pair of tentative edges, the output sequence has LDS at most two. \end{lemma}
\begin{proof} We give a proof by contradiction. A decreasing subsequence of length three must consist of a down-pair $a,b$ from one of the sequences and an element $e$ from the other sequence; $e$ is either a solitaire or part of a down-pair $c,d$ with $e \in \sset{c,d}$. 

Assume first that $e$ is a solitaire. If $e > a$, we have the edge $a \rightarrow e$, if $b > e$ we have the edge $e \rightarrow b$, and if $a > e > b$, we have the pair of tentative edges $b \dashrightarrow e \dashrightarrow a$. In either case, it is guaranteed that the three elements do not form a decreasing sequence of length three.
	
Assume next that $e$ is part of a down-pair $c,d$. We will argue that the generated edges guarantee that the four elements $a,b,c,d$ do not produce a decreasing sequence of length three in the output. We have the edges $a \rightarrow b$ and $c \rightarrow d$, $\min(a,c) \rightarrow \max(a,c)$ and $\min(b,d) \rightarrow \max(a,d)$. Assume w.l.o.g.\ that $a > c$. Then, we have the edge $c \rightarrow a$. If $b > d$, we also have the edge $d \rightarrow b$, and, hence, the two possible interleavings are $c,a,d,b$ and $c,d,a,b$. Since $c$ precedes $a$ and $d$ precedes $b$, a decreasing sequence can contain at most one of $a$ and $c$ and at most one of $b$ and $d$. Thus, its length is at most two. If $d > b$, we have the edge $b \rightarrow d$, and, hence, the only arrangement is $c,a,b,d$. Since $c$ precedes $a$ and $b$ precedes $d$, a decreasing sequence can contain at most one of $a$ and $c$ and at most one of $b$ and $d$. Thus, its length is at most two. 
\end{proof}

We next give an algorithm $\Gmrg$ for merging two sequences. It always considers the elements at the front of the two queues and decides which one to move to the output. It may also decide that the two sequences cannot be merged. If a front element is moved to the output, all edges incident to the element are deleted and the companions of all outgoing tentative edges are deleted. 
We call a front element of a sequence \emph{free} if it has no incoming edge, \emph{blocked} if is has an incoming strict edge, and \emph{half-free} if there is an incoming tentative edge but no incoming strict edge. Here are the rules.
\begin{enumerate}[a)]
	\item If only one of the front elements is free, move it. If both front elements are free, move the smaller one. 
	\item If both front elements are blocked, declare failure.
	\item If one front element is blocked and the other one is half-free, move the half-free and turn the companion edges of the incoming tentative edges into strict edges.
	\item If both front elements are half-free, move the smaller one. Turn the companion edges of the incoming tentative edges into strict edges.
\end{enumerate}

Figure~\ref{illustration} illustrates the algorithm. In the left example, there is no solitaire, and, hence, all edges are strict. The ordering $4, 7, 1, 2, 8, 9, 3, 6$ is a solution. Element 4 is free and is moved to the output. Then 7 is free, then 1, then 2, and so on. \newline
In the right example, element 5 is a solitaire and we have the pairs $1 \dashrightarrow 5 \dashrightarrow 7$, $3 \dashrightarrow 5 \dashrightarrow 7$ and $3 \dashrightarrow 5 \dashrightarrow 9$ of tentative edges. We proceed as above and move 4, 7, 1, and 2 to the output. When we move 7, the tentative edges $1 \dashrightarrow 5$ and $3 \dashrightarrow 5$ are turned into solid edges. When we move 1, the first of these edges is realized. However, after removing 2, 5 and 9 are blocked by the solid edges $3 \rightarrow 5$ and $8 \rightarrow 9$ respectively.

\begin{lemma}\label{lemma:completion} If $\Gmrg$ runs to completion, the output sequence obeys all strict edges and one of the tentative edges in each pair of tentative edges, and, hence, the output has LDS at most two. \end{lemma}
\begin{proof} All strict edges are obeyed since blocked elements are never moved to the output. Also, one edge of each pair of tentative edges is realized. This can be seen as follows. When we move an element to the output, all outgoing tentative edges are realized, and, hence, we may remove their companions. When we move a half-free element to the output, the incoming tentative edges are not realized. We compensate for this by making their companions strict. 
\end{proof}

\begin{lemma}\label{lemma:failure} If $\Gmrg$ fails, the two sequences cannot be merged.\end{lemma}
\begin{proof}  We assume that we start with two sequences that can be merged and consider the first moment of time when the algorithm makes a move that destroys mergeability. We will then consider all possible actions of the algorithm at this point and argue a contradiction.

Let $e$ and $f$ be the heads of the queues before the move of the algorithm. The algorithm moves $e$ to the output, but the only move that maintains mergeability is moving $f$. Then, $e$ is not blocked as the algorithm moves it.
	
	\begin{claim} $f$ is not blocked. \end{claim}
	\begin{claimproof} 
    Assume $f$ is blocked. Then, there is a strict edge $a \rightarrow f$ from an element $a$ in $e$'s queue. The strict edge was either created as a strict edge or as a tentative edge and later turned into a strict edge.
		
		Assume first that either $a$ or $f$ is a solitaire. If the edge was created as a strict edge, and $f$ is a solitaire, then $f > a$, and $a$ is the left element of a down-pair $a,b$. The decreasing sequence $f,a,b$ will be created, a contradiction. If the edge was created as a strict edge, and $a$ is a solitaire, then $f > a$, and $f$ is the right element of a down-pair $g,f$. The decreasing sequence $g,f,a$ will be created. If the edge was created as a tentative edge, and $f$ is a solitaire, then there is a down-pair $a',a$ in $e$'s input sequence with $a' > f > a$, and we had the tentative edges $a \dashrightarrow f \dashrightarrow a'$. Since the edge $a \dashrightarrow f$ is now a solid edge, $a'$ was already moved to the output. Thus, the decreasing sequence $a', f, a$ will be created. If the edge was created as a tentative edge, and $a$ is a solitaire, then there is a down-pair $f,f'$ in $f$'s input sequence with $f > a > f'$, and we had the tentative edges $f' \dashrightarrow a \dashrightarrow f$. Since $f'$ is after $f$ in $f$'s input sequence, the tentative edge $a \dashrightarrow f$ is still there as a tentative edge.
		
		We come to the case that $a$ and $f$ are elements of down-pairs. Then, either both are left elements or both are right elements in their down-pairs, 
        as we do not create edges (neither strict nor tentative) between left and right endpoints of down-pairs. Also, $a < f$ since edges always go from the smaller element to the larger element. 
        If both are left elements, and $a,b$ is a down-pair with $a$ on the left, the decreasing sequence $f,a,b$ will be created. If both are right elements, and $g,f$ is a down-pair with $f$ on the right, the decreasing sequence $g,f,a$ will be created.
		
		We have now shown that a blocked element never has to be moved to the output.\end{claimproof}

	So, we have to consider four cases, depending on which of $e$ and $f$ are free or half-free.
	
	Let $C = P\,f\,Q\,e\,R$ be a correct merging; here $P$, $Q$ and $R$ are sequences. We want to argue that $M = P\, e\, f\, Q\, R$ is also a correct merging. Assume otherwise. Then, there  must be a decreasing sequence of length 3 in $M$. As this decreasing sequence does not exist in $C$, it must involve $e$. All elements in $Q$ come from $f$'s queue, as no element from $e$'s queue can overtake $e$ in $C$.
	
	If the decreasing subsequence in $M$ does not involve $f$ (this is definitely the case if $e < f$), it must contain an element from $Q$, say $y$, as it would otherwise also exist in $C$. 
	
	If $f > e$, all elements in $P$ must be smaller than $f$, and all elements in $R$ must be larger than $e$, as $C$ would otherwise contain a decreasing subsequence of length three.

	\textbf{$e$ and $f$ are free:} We have $e < f$ since the algorithm moves $e$. So, the decreasing sequence cannot contain $f$. Also, $e$ cannot be the first element of the decreasing sequence as replacing the first element by $f$ would otherwise result in a decreasing sequence of length three in $C$. It cannot be the last element of the decreasing sequence, as the sequence contains an element of $Q$. So, the decreasing sequence has the form $x\, e\, y$ with $x$ in $P$ and $y$ in $Q$. Since $f > e > y$, $f,y$ is a down-pair in $f$'s sequence. If $e$ is a solitaire, we have the tentative edge $e \dashrightarrow f$ (Rule d3), and, hence, $f$ is not free, a contradiction. If $e$ is the left element in a down-pair, we have a strict edge $e \rightarrow f$ (Rule b1), a contradiction. If $e$ is the right element in a down-pair, we have the strict edge $y \rightarrow e$, a contradiction (Rule b1).
	
	\textbf{$e$ is free and $f$ is half-free:} Since $f$ is half-free, there is a tentative edge $d \dashrightarrow f$ from an element $d$ in $e$'s queue. Thus, $f$ is a solitaire or a left endpoint.
	
	If $f$ is a solitaire, there is a down-pair $c,d$ in $e$'s sequence with $c > f > d$; $d$ is still in $e$'s queue. 
	If $e < f$, the decreasing sequence has the form $x\, e\, y$ with $x$ in $P$ and $y$ in $Q$ by the same argument as in the first case. Recall that all elements in $Q$ come from $f$'s queue. Thus, $f,y$ is a down-pair, a contradiction to $f$ being a solitaire. 
	If $e > f$, $e,d$ is a down-pair with $e > f > d$. Thus, we have the tentative edge $f \dashrightarrow e$ (Rule d3), and $e$ is not free, a contradiction.
	
	If $f$ is a left endpoint, $d$ is a solitaire in $e$'s sequence. We have $f > d$. Since $d$ is a solitaire, we have $e < d$, and, hence, $e < f$. By the same argument as in the first case, the decreasing sequence has the form $x\, e\, y$ with $x$ in $P$ and $y$ in $Q$. Thus, $f,y$ is a down-pair with $f > d > e > y$. If $e$ is a left endpoint, we have the strict edge $e \rightarrow f$ (Rule b1), and $f$ is blocked. If $e$ is a right endpoint, we have the strict edge $y \rightarrow e$ (Rule b1), and $e$ is not free. If $e$ is a solitaire, we have the tentative edges $y \dashrightarrow e \dashrightarrow f$ (Rule d3), and $e$ is not free. In either case, we derived a contradiction. 
	
	\textbf{$e$ is half-free and $f$ is free:} This case cannot arise, as the algorithm would move $f$. 
	
	\textbf{$e$ and $f$ are half-free:} We have $e < f$, since the algorithm moves $e$. By the same argument as in the first case, the decreasing sequence has the form $x\, e\, y$ with $x$ in $P$ and $y$ in $Q$.
	
	$f,y$ is a down-pair in $f$'s sequence. If $e$ is the left element in a down-pair, we have a strict edge $e \rightarrow f$ (Rule b1), a contradiction. If $e$ is the right element in a down-pair, we have the strict edge $y \rightarrow e$ (Rule b1), a contradiction.
	
	If $e$ is a solitaire, all elements in $P$ that come from $e$'s queue are smaller than $e$, and all elements in $R$ are larger than $e$. This implies that $x$ comes from $f$'s queue and precedes $f$ in its input sequence, and, hence, $x,y$ is a down-pair in $f$'s queue with $x > e > y$. So, we have the tentative edges $y \dashrightarrow e \dashrightarrow x$ (Rule d3). When $x$ was moved to the output, the tentative edge $y \dashrightarrow e$ was converted to a strict edge, a contradiction.
\end{proof}

If there are no solitaires, the algorithm greatly simplifies. All edges of the precedence graph are strict, and we have mergeability if and only if the precedence graph is acyclic. A cycle in the precedence graph is a short certificate for non-mergeability. In the presence of solitaires, we have pairs of tentative edges and the question is whether there is a choice of tentative edges, one from each pair of tentative edges, such that the resulting graph is acyclic, or whether the graph is cyclic for any choice. We next show how to formulate this question as a 2-SAT problem. 

Let us use $q_1q_2 \ldots q_n$ and $p_1p_2\ldots p_m$ to denote the two input sequences. The edges defined by our rules b) and d) run between the two input sequences. For each such edge, we have a variable $x_e$: $x_e = \mathrm{true}$ if the output realizes $e$, and $x_e = \mathrm{false}$ otherwise. We say that two edges
$q_h p_i$ and $p_j q_k$ cross if and only if $i < j$ and $k < h$ or $i = j$ and $h = k$. If there are crossing strict edges, we have non-mergeability. We have the following clauses.
\begin{itemize}
	\item For each strict edge $e$, the singleton clause $x_e$.
	\item For each companion pair $e$ and $e'$ of tentative edges, $x_e \vee x_{e'}$ and $\neg x_e \vee \neg x_{e'}$
	\item For each pair of crossing edges $e$ and $e'$, $\neg x_e \vee \neg x_{e'}$.
\end{itemize}
Strict edges must be realized. From each pair of companion edges, exactly one has to be realized. From each pair of crossing edges, at most one should be realized. Let $F$ be the conjunction of the clauses defined above. 

\begin{lemma}\label{lemma:2sat} $F$ is satisfiable if and only if the two sequences are mergeable. \end{lemma}
\begin{proof} Assume $F$ is satisfiable. Let $E$ be the set of edges corresponding to true variables; $E$ includes all strict edges, exactly one out of each pair of tentative edges, and at most one edge in each pair of crossing edges. Add the edges $q_iq_{i+1}$ for $1 \le i < n$ and $p_j p_{j+1}$ for $1 \le j < m$. The resulting graph is acyclic. Assume otherwise, and let $q_i$ and $p_j$ be the elements with minimum index in a cycle. Then, there must be an edge $p_h q_i$ with $h > j$ and $q_k p_j$ with $k > i$ or with $i = k$ and $j = h$. So, we have a pair of realized crossing edges which is impossible.
	
	Assume, conversely, that the sequences are mergeable. Then all strict edges are realized, and exactly one out of each companion pair of tentative edges is realized. Set the variables corresponding to realized edges to true. Then, the first two types of clauses are certainly true. The third type is also true since in a pair of crossing edges, at most one can be realized.
\end{proof}

For 2-SAT formulae, there are simple witnesses of satisfiability and non-satisfiability~\cite{Apsvall}, computable in polynomial time. We briefly review the argument. A satisfying assignment is a witness for satisfiability. For a witness of non-satisfiability, write all clauses involving two literals as two implications, i.e., replace $x \vee y$ by $\neg x \rightarrow y$ and $\neg y \rightarrow x$. We use $\rightarrow^*$ to denote the transitive closure of $\rightarrow$. Clearly, all singleton clauses need to be satisfied. Do so and follow implications.
If this leads to a contradiction, the formula is not satisfiable, and we have a witness of the form $x \rightarrow^* y$ and $x' \rightarrow^* \neg y$ for two singleton clauses $x$ and $x'$ and a literal $y$. After this process, there are no singleton clauses. Create a graph with the literals as the nodes of the graph and the implications as the directed edges. If there is a cycle containing a variable and its complement, the formula is not satisfiable. So, assume otherwise.
Consider any variable $x$. We may have the implication $x \rightarrow^* \neg x$ or the implication $\neg x \rightarrow^* x$, but we cannot have both, as we would have a cycle $x \rightarrow^* \neg x \rightarrow* x$ involving $x$ and $\neg x$. If we have only the former, set $x$ to false, if we have only the latter, set $x$ to true, and if we have neither, set $x$ to an arbitrary value. Follow implications. This cannot lead to a contradiction. Assume we have the implication $x \rightarrow^* \neg x$ and, hence, set $x$ to false and $\neg x$ to true. If this led to a contradiction, we would have implications $\neg x \rightarrow^* y$ and $\neg x \rightarrow^* \neg y$ for some $y$, and, hence, $y \rightarrow* x$, and, hence, $\neg x \rightarrow^* y \rightarrow^* x \rightarrow^* \neg x$. The argument for the alternative case is symmetric.

Figure~\ref{illustration} illustrates the algorithm and the connection to 2-SAT. In the 2-SAT problem we have to realize all solid edges. Focus on the example on the right. Since the solid edge $1 \rightarrow 2$ crosses the tentative edge $5 \dashrightarrow 7$, the tentative edge $3 \rightarrow 5$ has to be realized. However, it crosses $8 \rightarrow 9$, and, hence, the formula cannot be satisfied.

Figure~\ref{illustration2} gives a more substantial example. Here, the lower sequence contains three solitaires, namely 5, 8, and 12. The solitaire $5$ is enclosed by the down-pair $(6,3)$, and, hence, we have the tentative edges $3 \dashrightarrow 5 \dashrightarrow 6$. The solitaire $8$ is enclosed by the down-pair $(10,7)$, and, hence, $7 \dashrightarrow 8 \dashrightarrow 10$. The solitaire 12 is enclosed by $(13,11)$, and, hence, $11 \dashrightarrow 12 \dashrightarrow 13$. We also have the strict edges $1 \rightarrow 2$ and $15 \rightarrow 16$. The two sequences are not mergeable. Since we have the strict edge $1 \rightarrow 2$, we cannot realize $5 \dashrightarrow 6$. Thus, we must realize $3 \dashrightarrow 5$ and, hence, cannot realize $8 \dashrightarrow 10$. Thus, we must realize $7 \dashrightarrow 8$ and, hence, cannot realize $12 \dashrightarrow 13$. Thus, we must realize $11 \dashrightarrow 12$. This conflicts with the strict edge $15 \rightarrow 16$. 
\newline
Without the element 14, the two sequences can be merged into $4,6,1,2,10,3,5,13,7,8,16,11,12,15$. Instead of the solid edge $15 \rightarrow 16$, we have the tentative edges $11 \dashrightarrow 15 \dashrightarrow 16$. The former is realized.

We have now shown that mergeability to a sequence with LDS two can be decided in polynomial time. In the next section, we will establish linear time. 

\begin{figure}[t]
\centering
		\includegraphics[width=\textwidth]{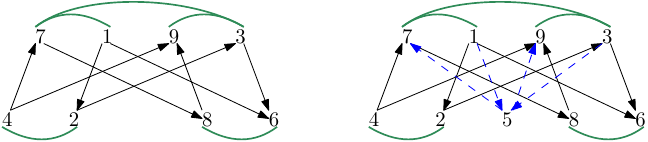}
	\caption{\label{illustration} The green brackets indicate down-pairs, the solid arrows strict edges, and the dashed arrows tentative edges
	}
\end{figure}

\begin{figure}[t]
\centering
		\includegraphics[width=0.5\textwidth]{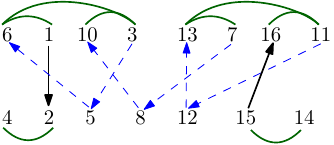}
	\caption{\label{illustration2} The green brackets indicate down-pairs, the solid arrows strict edges, and the dashed arrows tentative edges. Not all edges are drawn in the figure}
\end{figure}

\subsection{Merging in Linear Time}\label{Linear Time}

We give a linear time implementation of the algorithm in the preceding section. In particular, we show that there is no need to construct the precedence graph, but that the decision of which front element to move to the output can be made locally after a small amount of precomputation. This is the content of Lemma~\ref{mergability refined}. After the Lemma, we show how to compute the information required for the refined rules in linear time. Finally, we put all the pieces together and prove Theorem~\ref{intro:mergeability}.

\begin{lemma}\label{mergability refined} Let $c_i$ be the head of the $i$-th queue; here queue means the suffix of elements that have not been dequeued yet. 
\begin{enumerate}
\item If $c_i$ is smaller than all elements in the other queue, it is free.
\item If neither $c_1$ nor $c_2$ is a left element, the smaller of the two is free. 
\item If a $c_i$ is a left element, it is blocked iff the other queue contains a smaller left element. If there is such a left element in the other queue, it is the leftmost left element. 

If there is no such left element in the other queue, $c_i$ is half-blocked iff the other queue contains a solitaire $s$ such that $c_i > s > r_i$, where $r_i$ is the first right element after $c_i$. 
\item We now assume that at least one first element is a left element, and that both queues contain an element smaller than the first element in the other queue. 
\begin{enumerate}
\item If both first elements are left elements, the larger is blocked, and the smaller is non-blocked. 
\item Assume the first queue starts with a left element $\ell_1$, and the second queue starts with either a right element or a solitaire. Let $r_1$ be the first right element after $\ell_1$. Then, $\ell_1 > c_2 > r_1$. 
\item If one first element is a left element and the other a right element, the right element is blocked.
\item If one first element is a left element, the other is a solitaire, and a left element larger than the solitaire has already been dequeued, the solitaire is blocked. 
\item If one first element is a left element, the other is a solitaire, and no left element larger than the solitaire was already dequeued, the solitaire is half-blocked and smaller than the left element.  \end{enumerate}
\item The above covers all cases.
\end{enumerate}
\end{lemma}
\begin{proof}
\begin{enumerate}
\item All constraint edges between sequences go from smaller to larger elements. 
\item If $c_i$ is a right element or a solitaire, all elements after $c_i$ in the $i$-th queue are larger than $c_i$. Thus, if neither first element is a left element, the smaller first element is free. 
\item A solid incoming edge into a left element can only come from a smaller left element. The leftmost left element is the smallest left element in the other queue. 

A tentative incoming edge can only come from a solitaire with $c_i > s > r_i$, where $r_i$ is a right element in $c_i$'s queue. If there is such an $r_i$, it is the first right element. 
\item 
\begin{enumerate}
\item We have the strict edge $\min(\ell_1,\ell_2) \rightarrow \max(\ell_1,\ell_2)$, and, hence, the larger is blocked. The smaller is either free or half-blocked. 
\item Assume the first queue starts with a left element $\ell_1$. There must be a right element $r_1$ after it in the first queue. $r_1$ is the smallest element in the first queue, and, hence, $c_2 > r_1$. If the first element in the second queue is either a solitaire or a right element, it is the smallest element in the second queue. Thus, $\ell_1 > c_2$. 
\item Assume $c_1$ is a left element and $c_2$ is a right element. Let $r_1$ be the first right element in the first queue. Then, $r_2 > r_1$ by the preceding item, and, hence, we have a strict edge from $r_1$ to $r_2$. Hence, $r_2$ is blocked. 
\item Assume the first queue starts with a left element. Then, $\ell_1 > s_2 > r_1$, where $r_1$ is the first right element in the first queue. If a left element $\ell'_1$ with $\ell'_1 > s_2$ was already dequeued, the tentative edge $s_1 \dashedrightarrow \ell'_1$ was not realized, and, hence, the tentative edge $r_1 \dashedrightarrow s_1$ became solid. 
\item Assume the first queue starts with a left element. Then, $\ell_1 > s_2 > r_1$, where $r_1$ is the first right element in the first queue. All left elements already dequeued from the first queue are smaller than $s_2$, and, hence, all tentative edges $r \dashedrightarrow s_1$ from a right element in the first queue into $s_2$ are still tentative. We also have the tentative edges $r_1 \dashedrightarrow s_1 \dashedrightarrow \ell_1$. Thus, $s_1$ is half-blocked.  
\end{enumerate}
\item If neither first element is a left element, the smaller is free and moved to the output (item 2). If a first element is a left element, item 3 tells  whether the element is free, half-blocked, or blocked. If both first elements are left elements, item 4a tells which one to move. If one first element is a left element, and the other is either a right element or a solitaire, items 4c, 4d, and 4f tell the status of the other element. 
\end{enumerate}
\end{proof}

What is needed to implement the preceding Lemma?

\begin{itemize}
\item Each element needs to know its type: left or right element or solitaire. 
\item For each suffix of a queue, we need to know the smallest right element in the suffix, the smallest left element in the suffix, and the smallest solitaire in the suffix.  
\item The largest already dequeued left element. 
\end{itemize}

\begin{claim}\label{claim:const time} If the information above is available, all decisions in Lemma~\ref{mergability refined} can be made in constant time. The above information can be precomputed in linear time. \end{claim}
\begin{proof} We go through the various items of Lemma~\ref{mergability refined} and argue that the decision can be made based on the information above. For the decisions in cases 1, 2, 3, 4a, 4b, 4c, and 4d, the information of the first two items suffices. For the case distinction 4e or 4f, one also needs the third item. 

Item 3 needs an additional explanation. If $c_i$ is a left element, and we have $c_i > s > r_i$ for some solitaire in the other queue, and $r_i$ is the first right element in $c_i$'s queue, $s$ is not necessarily the first solitaire in the other queue. However, if the first solitaire in the other queue is smaller than $c_i$ and $r_i$, then it is smaller than all elements in $c_i$'s queue. Thus, the solitaire is free, and it does not matter, whether $c_i$ is free or half-blocked. 

An element is a right element if there is a larger element to the left of it. We scan the sequence from left to right and maintain the maximum. Left elements are determined by a right to left scan. Elements that are neither left nor right are solitaires. For the second item, we do a right to left scan, and the third item is kept during the execution. \end{proof}

The algorithm above either produces an output of LDS 2 or stops because both first elements are blocked. Can we in the latter case at least produce an output of LDS 3? Yes, we simply continue with the smallest element first strategy. 

\begin{claim}\label{claim:fail lds 3} If both first elements are blocked, continuing with the smallest element first strategy will produce an output of LDS 3. 
\end{claim}
\begin{proof} It is clear that smallest element first runs to completion. So, we only have to show that the output has LDS 3. We first show that among two left elements, the smaller is always output before the larger. This is clear if they come from the same queue, as the smaller precedes the larger. If they come from different queues, and one of them is output before the algorithm blocks, this holds true, because we have a strict edge from the smaller to the larger. So, assume that both are output by smallest element first strategy. Let $\ell_1$ and $\ell_2$ be two left elements with, say, $\ell_1 < \ell_2$, and assume that $\ell_2$ is output before $\ell_1$ by smallest element first. Consider the situation just before $\ell_2$ is output. All elements preceding $\ell_1$ in its queue are smaller than $\ell_1$, and, hence, the smallest element first strategy  will output them, and $\ell_1$ before $\ell_2$. 

So, a decreasing sequence can involve at most one left element. It can also involve at most one solitaire. So, if it has length four, it must involve two right elements, one from each queue. One of them must come from the same queue as the solitaire. But this is impossible, because solitaires are larger than all preceding elements and smaller than all succeeding elements. 
\end{proof}

Theorem~\ref{intro:mergeability} now follows from Lemma~\ref{mergability refined} and Claims~\ref{claim:const time} and~\ref{claim:fail lds 3}.

\section{Minimizing Down-Steps}\label{DownSteps}

In this section we switch focus to the number a down-steps, a different measure for ``sortedness''.
We present an algorithm $\Opt$ for minimizing the number of down-steps in a sequence using $k$ queues. This algorithm runs in time $O(n \log k)$, where $n$ is length of the sequence. The algorithm is online and optimal in a strong sense. Firstly, no algorithm can produce an output with fewer down-steps than $\Opt$ does. Secondly, if the input has $D$ down-steps, the output has at most $\floor{D/k}$ down-steps. Thirdly, if an online algorithm deviates from $\Opt$, it is suboptimal, i.e., the input sequence can be extended so that the algorithm produces an output with a strictly larger number of down-steps.

The result sheds some light on whether having several multilane buffers in a row is better than having all lanes in parallel in a single buffer. For example, with 6 queues in parallel, the number of down-steps can be reduced to $\floor{D/6}$ and with three 2-line buffers in a row, the number of down-steps can be reduced to $\floor{\floor{\floor{D/2}/2}/2}$. The latter number is generally smaller. We come back to this example in the conclusion.
 
The algorithm is as follows: Recall that we have $k$ queues, and our input is a sequence of length $n$. For simplicity, we assume that all numbers in the input are non-negative. 
We define the enqueuing and dequeuing strategy of our algorithm $\Opt$ and then prove its optimality.

  \paragraph*{Enqueuing:} The queues are numbered $0$ to $k - 1$. We denote the last element in queue $i$ by $\ell_i$. Initially, when all queues are empty, $\ell_i$ has the fictitious value $-i-1$. One of the queues is the \emph{base queue}; we use $b \in [k-1]$ to denote the base queue. Initially, $b = 0$. We maintain the following \emph{invariant} at all times: 
  \[ \ell_b > \ell_{b+1} > \ldots > \ell_{b + k - 1}.\]
  Indices are to be read modulo $k$, and the base queue is always the queue, where the numbering of the queues starts. Note that the fictitious values are chosen such that the invariant holds initially.

  Assume now that element $e$ is to be enqueued. Let $i \in \sset{0,1,\ldots, k-1}$ be minimal (if any) such that $e > \ell_{b + i}$. If $i$ exists, either $i = 0$ and $e > \ell_b$, or $i > 0$ and $\ell_{b+i - 1} > e > \ell_{b + i}$.
  \begin{itemize}
  \item If $i$ exists, we append $e$ to queue $b +i$. This does not create a down-step and maintains the invariant.
  \item If $i$ does not exist, i.e., $\ell_{b + k - 1} > e$, we append $e$ to queue $b$ and increase $b$ by 1. This creates a down-step in queue $b$ unless $e$ is its first element. The invariant is maintained.
  \end{itemize}

  A \emph{run} is an increasing sequence. At the end of the enqueuing, all queues contain an equal number of runs up to one. Note that the down-steps are generated in round-robin fashion. The first down-step is created in queue $0$, the second in queue $1$, and so on. Let $C$ be the number of times $b$ is increased. We claim $C \le D$, where $D$ is the number of down-steps in the input sequence. Indeed, an input element $e$ that causes a change of base queue is smaller than the last element of all queues. Since the element preceding $e$ in the input is the last element of some queue, a change of base queue can only be caused by a down-step in the input sequence. Thus $C \le D$. The number of runs in queue zero is 1 if $C < k$ and increases by one for every $k$ increases of $C$. Thus the number of runs in queue zero is at most $\floor{C/k} + 1$, and the number of down-steps in queue zero is one less. We conclude that the number of down-steps in any queue is at most $\floor{C/k}$, which in turn is at most $\floor{D/k}$.

  The enqueuing strategy is inspired by Patience Sort~\cite{aldous1999longest}. Patience Sort sorts a sequence of numbers using a minimum number of queues. It is as above with one difference: if $e$ cannot be appended to an existing run, then it opens up a new queue.

  \paragraph*{Dequeuing:} We first merge the first runs in all queues into a single run, then merge the second runs in all queues, and so on. So the number of runs in the output sequence is the maximum number of runs in any of the queues.

  It is clear that the algorithm can be made to run in time $O(n \log k)$. We keep the values $\ell_i$ in an array $A$ of size $k$.  When an element $e$ is to be enqueued, we perform binary search on the sorted array $A[b .. b + k - 1]$; again indices are to be read modulo $k$.

  In the next section, we prove the optimality of $\Opt$ and in 
Section~\ref{Sequence Decomposition}, we characterize the runs generated by the algorithm in the different queues without reference to the algorithm.

\subsection{Optimality and Uniqueness of {\texorpdfstring{$\Opt$}{Opt}}}
\label{Transformation}

Let $A$ be any algorithm. We will show that the queue contents constructed by $A$ can be transformed into the contents constructed by $\Opt$ without increasing the number of down-steps. Dequeuing will be considered later. 

Let $t$ be minimal such that $A$ and $\Opt$ differ, i.e, $A$ enqueues $e = e_t$ into queue $i$, and $\Opt$ enqueues $e$ into queue $j$ with $i$ different from $j$. We concentrate on the $i$-th and the $j$-th queue. Let $a_i$ and $a_j$ be the contents of these queues just before enqueuing $e_t$, and let $s_i$ and $s_j$ be their continuations by $A$; $s_i$ starts with $e$.
We swap $s_i$ and $s_j$, i.e, the contents of the $i$-th queue become $a_i s_j$, and the contents of the $j$-queue becomes $a_j s_i$. The number of down-steps stays the same except maybe for the down-steps at the borders between the $a$-parts and the $s$-parts.

We use $\ell(a_i)$, $\ell(a_j)$, $f(s_i)$, and $f(s_j)$ to denote the last and first elements of $a_i$, $a_j$, $s_i$ and $s_j$, respectively. We have $e = f(s_i)$. We need to show that the number of down-steps in the pairs $(\ell(a_j), f(s_i))$ and $(\ell(a_i),f(s_j))$ is no larger than in the pairs $(\ell(a_i), f(s_i))$ and $(\ell(a_j),f(s_j))$.

We distinguish cases according to the action of $\Opt$. If $e$ is smaller than the last elements of all existing parts, $\Opt$ incurs a down-step and adds $e$ to the queue with largest last element. Otherwise, $\Opt$ appends $e$ to the queue whose last element is largest among the last elements smaller than $e$.

\begin{description}
\item[$e$ is smaller than all last elements:] Then all $k$ queues have a non-fictitious last element since $e$ is a non-negative number, and fictitious last elements are negative. Thus, $\ell(a_j) > \ell(a_i) > e$, where the first inequality follows from $i \not= j$, and the fact that the $j$-th queue has the largest last element, and the second inequality holds because $e$ is smaller than \emph{all} last elements. Thus, we have a down-step in $a_i s_i$ and in $a_j s_i$. Let $e' = f(s_j)$. If we have no down-step in $a_j s_j$, then we have no down-step in $a_i s_j$, and, hence, the swap does not increase the cost. If we have a down-step in $a_j s_j$, then $A$ incurred two down-steps before the swap, and, hence, the swap cannot increase the cost. The swap decreases the number of down-steps if $\ell(a_j) > e' > \ell(a_i)$.
\item[$e$ is larger than some last element:]
We have $\ell(a_j) < e = f(s_i)$ because $\Opt$ enqueues after the largest last element that is smaller than $e$. Thus,
we have no down-step at the border from $a_j$ to $s_i$, and, hence, $A$ incurs at most one down-step after the swap. So, we only need to show that if the number of down-steps at the borders is zero before the swap, then it is zero after the swap.

Assume the number is zero before the swap. Then, $\ell(a_i) < e = f(s_i)$ and $\ell(a_j) < f(s_j)$. Also, $\ell(a_i) < \ell(a_j) < e$ since $\Opt$ enqueues $e$ into the queue with largest last element smaller than $e$. Thus, $\ell(a_i) < f(s_j)$, $\ell(a_j) < f(s_i)$, and no down-step will be introduced by the swap. 
\end{description}

We summarize: We have shown how to convert the enqueuing of any algorithm into the enqueuing of $\Opt$ without increasing the total number of down-steps. Moreover, $\Opt$ enqueues such that the number of down-steps in any two queues differs by at most one. Finally, the number of down-steps in the output is equal to the maximum number of down-steps in any queue. Thus, every algorithm must generate at least as many down-steps in the output as $\Opt$.

We will next show that $\Opt$ is the unique optimal \emph{online} algorithm. Recall that in an online algorithm the input is presented as a sequence $e_1,e_2. \ldots, e_j, \ldots$, and we have to pick a queue for element $e_j$ based on the current state of the queues and $e_j$. The processing of $e_j$ is independent of all the elements that follow it in the sequence.

\begin{lemma}
\label{lemma:downuniqu}
$\Opt$ is the unique optimal online algorithm.
\end{lemma}
\begin{proof} 
 Consider any other online algorithm $A$. We will show that if $A$ at some point deviates from $\Opt$, then there is a continuation of the input that will force $A$ to be worse.

 Let $e_1,e_2,\ldots,e_j,\ldots$ be the input sequence. When an element $e_j$ is to be enqueued, an online algorithm chooses the queue based on the value of $e_j$ and the current contents of the queues, but independently of the continuation $e_{j+1}, \ldots$. Assume now that $A$ is not identical to $\Opt$. Then, there is an input sequence $e_1,e_2,\ldots,e_j$ such that $A$ and $\Opt$ act the same up to element $e_{j-1}$ but choose different queues for $e_j$. We will show that there is a continuation $e_{j+1}, \ldots$ that forces $A$ to incur more down-steps than $\Opt$.

Let us first assume that all queues are in use when $e$ in enqueued. We may assume $b = 0$. So, $\ell_0 > \ell_1 >  \ldots > \ell_{k-1}$, and these are also the tails of the queues for algorithm $A$. An element $e$ arrives.

Assume first that there is an $i$ such that $i = 0$ and $e > \ell_1$ or $\ell_{i-1} > e > \ell_i$. Algorithm $\Opt$ adds $e$ to the $i$-th queue, and algorithm $A$ adds $e$ to the $j$-th queue with $j \not= i$.

If $j < i$, $A$ incurs a down-step, and we have an input sequence on which $A$ does worse.

If $j > i$, we have to work harder. The final elements of the queues are now:
\begin{align*} \text{alg.\ $\Opt$}\quad&  \ell_0 > \ldots > \ell_{i-1} > e > \ell_{i+1} >\ldots > \ell_{j-1} > \ell_j >  \ell_{j+1} \ldots > \ell_{k-1};\\
  \text{alg.\ A}\quad&  \ell_0 > \ldots > \ell_{i-1} > e > \ell_i  > \ell_{i+1} > \ldots > \ell_{j-1} > \ell_{j+1} > \ldots > \ell_{k-1}.
\end{align*}
We now insert $k - 1 - i$ elements that algorithm $\Opt$ can insert without incurring a down-step, but A cannot. More precisely, we insert the decreasing sequence $\ell_{i+1} + \epsilon$, \ldots, $\ell_j + \epsilon$, \ldots $\ell_{k-1}+ \epsilon$, where $\epsilon$ is an infinitesimal. Algorithm $\Opt$ can insert these elements into queues $i+1$ to $k-1$ without incurring a down-step. Algorithm A either has to put two of these elements in the same queue or one element in one of the first $i+1$ queues. In either case, it incurs a down-step.

We come to the case $e < \ell_k$. Algorithm $\Opt$ puts $e$ into queue 0, and algorithm A puts $e$ into queue $j > 0$. The final elements of the queues are now:
\begin{align*} \text{algorithm $\Opt$}\quad&  \ell_1 > \ell_2 >\ldots > \ell_{j-1} > \ell_j >  \ell_{j+1} \ldots > \ell_{k-1} > e;\\
  \text{algorithm A}\quad&  \ell_0 > \ldots > \ell_{j-1} > \ell_{j+1} > \ldots > \ell_{k-1} > e.
\end{align*}
We next insert $\ell_j +\epsilon$, $\ell_{j+1} + \epsilon$, \ldots, $\ell_{k-1} + \epsilon$. Algorithm $\Opt$ inserts these elements without incurring a  down-step, but algorithm A must incur a down-step. The argument is as above.

We have now handled the situation when $A$ differs from $O$ after both algorithms fill all queues. We next deal with the situation, where only $h < k$ queues are used, and we have $\ell_0 >  \ldots > \ell_{h-1}$.
An element $e$ is added.
If $e > \ell_{h-1}$, we argue as above. If $e < \ell_{h-1}$ and A does not open a new queue, it incurs a down-step. If A opens a new queue, it does the same as $\Opt$.
\end{proof}

We have now completed the proof of Theorem~\ref{intro:downsteps}.

\subsection{Sequence Decomposition}\label{Sequence Decomposition}

In this section, we characterize the set of runs determined by algorithm $\Opt$. For patience sort such a characterization is easy. Queue zero contains all left-to-right maxima of the input sequence, queue one contains the left-to-right maxima of the sequence obtained by deleting the elements in queue zero, and so on. Alternatively, queue $\ell$ contains all elements of rank $\ell + 1$, where the rank of an element is the length of the longest decreasing subsequence ending in it.

Let $I_0$ be the input sequence. We will define a run $s_0$ in $I_0$ and obtain a modified input sequence $I_1$ by deleting $s_0$ from $I_0$. We will define a run $s_1$ in $I_1$ and obain a modified input sequence by deleting it from $I_1$. Generally, we will define a run $s_j$ in $I_j$ and obtain $I_{j+1}$ by deleting it. Assume we have constructed $I_j$. Intially, $j = 0$. 

Let $f_j$ be the leftmost element in $I_j$ in which a decreasing sequence of length $k+1$ ends. If there is no such element, let $f_j$ be a fictitious element after the end of $I_j$.
Let $P_j$ be the prefix of $I_j$ ending just before $f_j$, and let $s_j$ be the sequence of left-to-right maxima of $P_j$, i.e., $s_j$ starts with the first element of $P_j$ and is then always extended by the first element that is larger than its last element.

\begin{theorem} Algorithm $\Opt$ constructs runs $s_i$, $s_{i+k}$, $s_{i + 2k},\dots$ in queue $i$, $0 \le i < k$. Exactly the non-fictitious elements $f_j$ create down-steps. \end{theorem}
\begin{proof} We will show by induction on the number of elements already processed. Let $j$ be the number of times $b$ was increased. Then $b = j \bmod k$. Then the last run in queue $q_{b + \ell}$, $0 \le \ell < k$ consists of all already processed elements of $I_j$ of rank $\ell + 1$ with respect to $I_j$.
Initially, $j = 0$, $b = 0$, and no element is processed, and all queues are empty. So the claim holds. 

Let $e$ be the next element of the input, and let $r$ be its rank with respect to $I_j$.  

If $r \le k$, $e$ is added to queue $q_{b + r - 1}$. This can be seen as follows. Assume $e$ is added to queue $b + \ell$, where $\ell$ is smallest such that the last element of queue $b + \ell$ is smaller than $e$. Since the last elements of the queues form a decreasing sequence, we have $r \ge \ell + 1$.  On the other hand, the last run in each queue can contain at most one element of a decreasing sequence ending in $e$, and hence $\ell + 1 \le r$. We have now reestablished the induction hypothesis. 

If $r = k + 1$, $s_j$ is the current contents of $q_b$, $I_{j+1}$ is obtained from $I_j$ by deleting $s_j$, $e$ is enqueued into $q_b$ and $b$ is increased by one. The rank of the elements in $_{j+1}$ up to and including $e$ is one less than their rank in $I_j$. We have again reestablished the induction hypothesis. 
\end{proof}

\section{Conclusion}\label{conclusion}

Our research was motivated by a problem in car manufacturing, yet we considered an abstract and idealized version of a lane buffer, e.g., unbounded capacity per lane, to investigate foundational theoretical properties. Nevertheless, it is fair to ask whether our insights also lead to advice for the practitioner. To this end, consider a typical layout of a lane buffer with $k=2p$ parallel lanes as sketched in Figure~\ref{fig:layouts}a for $p=3$. We have shown in this paper that this variant will reduce the number of down-steps of an input sequence with $D$ down-steps to $\floor{D/k}$ or less. However, if we pair the lanes and alternate their directions from one pair to the other, we can obtain a series of $p$ buffers with two lanes each within the same bounding box as shown in Figure~\ref{fig:layouts}b, where we feed the output of two parallel queues again into two queues, and so on. Hence, the number of down-steps is reduced to $\floor{\floor{\ldots\floor{D/2}\ldots /2}/2}$, where we have a nesting of $p$.  More concisely, the number of down-steps is $\floor{D/2^p}$. The latter number of down-steps is generally much smaller. A disadvantage of this layout is that paths through the buffer get longer. However, it is not necessary to remove the connections from the beginnings to the ends of lanes in alternating direction, so that we can use such segments as bypasses as shown in Figure~\ref{fig:layouts}c to reduce the length of a shortest path through the buffer to the same length as in Figure~\ref{fig:layouts}a.

\begin{figure}[htp]
\resizebox{\textwidth}{!}{ %
\begin{tikzpicture}[>=stealth]

\begin{scope}[xshift=0cm]
  \draw[->] (-1,0) -- (-0.5,0);
  \draw[-] (-0.5,0) -- (0,0);
  \foreach \x in {0, 1, 2, 3, 4} {
    \draw[-] (\x,0) -- (\x+1,0);
    \draw[->] (\x+0.5,0) -- ++(0.1,0);
    \draw[-] (\x,-3) -- (\x+1,-3);
    \draw[->] (\x+0.5,-3) -- ++(0.1,0);
  }
  \draw[->] (5,-3) -- (5.5,-3);

  \foreach \x in {0, 1, 2, 3, 4, 5} {
    \draw[->] (\x,0) -- (\x,-1.5);
    \draw[-] (\x,-1.5) -- (\x,-3);
  }

  \node at (3,-4) {(a)};
\end{scope}

\begin{scope}[xshift=7cm]
  \draw[->] (-1,0) -- (-0.5,0);
  \draw[-] (-0.5,0) -- (0,0);
  \foreach \x in {0, 2, 4} {
    \draw[-] (\x,0) -- (\x+1,0);
    \draw[->] (\x+0.5,0) -- ++(0.1,0);
    \draw[-] (\x,-3) -- (\x+1,-3);
    \draw[->] (\x+0.5,-3) -- ++(0.1,0);
  }
  \foreach \x in {1} {
    \draw[-] (\x,-3) -- (\x+1,-3);
    \draw[->] (\x+0.5,-3) -- ++(0.1,0);
  }
  \foreach \x in {3} {
    \draw[-] (\x,0) -- (\x+1,0);
    \draw[->] (\x+0.5,0) -- ++(0.1,0);
  }
  \draw[->] (5,-3) -- (5.5,-3);

  \foreach \x in {0, 1, 4, 5} {
    \draw[->] (\x,0) -- (\x,-1.5);
    \draw[-] (\x,-1.5) -- (\x,-3);
  }
  \foreach \x in {2, 3} {
    \draw[->] (\x,-3) -- (\x,-1.5);
    \draw[-] (\x,-1.5) -- (\x,0);
  }

  \node at (3,-4) {(b)};
\end{scope}

\begin{scope}[xshift=14cm]
  \draw[->] (-1,0) -- (-0.5,0);
  \draw[-] (-0.5,0) -- (0,0);
  \foreach \x in {0, 1, 2, 3, 4} {
    \draw[-] (\x,0) -- (\x+1,0);
    \draw[->] (\x+0.5,0) -- ++(0.1,0);
    \draw[-] (\x,-3) -- (\x+1,-3);
    \draw[->] (\x+0.5,-3) -- ++(0.1,0);
  }
  \draw[->] (5,-3) -- (5.5,-3);

  \foreach \x in {0, 1, 4, 5} {
    \draw[->] (\x,0) -- (\x,-1.5);
    \draw[-] (\x,-1.5) -- (\x,-3);
  }
  \foreach \x in {2, 3} {
    \draw[->] (\x,-3) -- (\x,-1.5);
    \draw[-] (\x,-1.5) -- (\x,0);
  }

  \node at (3,-4) {(c)};
\end{scope}
\end{tikzpicture} 
} %
\vspace{-0.5cm}
\caption{Layout variations of a lane buffer covering the same space on the floor shop. The items enter on the top left and leave on the bottom right. }\label{fig:layouts}
\end{figure}

We conclude with some open problems:
\begin{itemize}
    \item 
    Our results depend on the fact that the queues are used in parallel. As discussed above, we can improve by pairing queues and use them sequentially. It is an interesting open problem to consider the general setting, where the queues are arranged in an arbitrary directed acyclic network. Is pairing the queues, using the pairs sequentially, and adding bypasses optimal?

    \item Theorem~\ref{intro:mergeability} only applies to two sequences of LDS 2. Can we generalize this result to multiple queues and larger LDS?

    \item Besides the two measures of disorder LDS and the number of down-steps considered here, it will be interesting to develop other measures and understand their algorithmic properties.

    \item Stacks, which obey the Last-In-First-Out(LIFO) principle, have also been considered for sorting sequences~\cite{tarjan1972sorting,golumbic2004algorithmic}. It is an interesting problem to characterize the sorting power of $k$ parallel stacks, for any fixed number $k$.
\end{itemize}

\section*{Declarations}
\bmhead{Funding}
This work was supported by the European Regional Development Fund (ERDF).

\bmhead{Competing interests}
The authors have no relevant financial or non-financial interests to disclose.

\bibliography{ref,local,lib}

\begin{thebibliography}{19}
\providecommand{\natexlab}[1]{#1}
\providecommand{\url}[1]{{#1}}
\providecommand{\urlprefix}{URL }
\providecommand{\doi}[1]{\url{https://doi.org/#1}}
\providecommand{\eprint}[2][]{\url{#2}}
 \bibcommenthead

\bibitem[{Aldous and Diaconis(1999)}]{aldous1999longest}
Aldous D, Diaconis P (1999) Longest increasing subsequences: from patience
  sorting to the {Baik-Deift-Johansson} theorem. Bulletin of the American
  Mathematical Society 36(4):413--432

\bibitem[{Aspvall et~al.(1979)Aspvall, Plass, and Tarjan}]{Apsvall}
Aspvall B, Plass M, Tarjan R (1979) A linear-time algorithm for testing the
  truth of certain quantified boolean formulas. Information Processing Letters
  8:121–123

\bibitem[{Boge and Knust(2020)}]{Boge-Knust}
Boge S, Knust S (2020) The parallel stack loading problem minimizing the number
  of reshuffles in the retrieval stage. Eur J Oper Res 280(3):940--952.
  \doi{10.1016/j.ejor.2019.08.005},
  \urlprefix\url{https://doi.org/10.1016/j.ejor.2019.08.005}

\bibitem[{Boroujeni et~al.(2021)Boroujeni, Ghodsi, and Seddighin}]{Boroujeni}
Boroujeni M, Ghodsi M, Seddighin S (2021) Improved {MPC} algorithms for edit
  distance and {Ulam} distance. IEEE Transactions on Parallel and Distributed
  Systems 32(11):2764--2776. \doi{10.1109/TPDS.2021.3076534}

\bibitem[{Boysen and Emde(2016)}]{Boysen:stack-blockage}
Boysen N, Emde S (2016) The parallel stack loading problem to minimize
  blockages. European Journal of Operational Research 249(2):618--627.
  \doi{https://doi.org/10.1016/j.ejor.2015.09.033},
  \urlprefix\url{https://www.sciencedirect.com/science/article/pii/S0377221715008759}

\bibitem[{Chandramouli and Goldstein(2014)}]{Chandramouli}
Chandramouli B, Goldstein J (2014) Patience is a virtue: Revisiting merge and
  sort on modern processors. In: Proceedings of the 2014 ACM SIGMOD
  International Conference on Management of Data. Association for Computing
  Machinery, New York, NY, USA, SIGMOD '14, p 731–742,
  \doi{10.1145/2588555.2593662},
  \urlprefix\url{https://doi.org/10.1145/2588555.2593662}

\bibitem[{Cormen et~al.(2022)Cormen, Leiserson, Rivest, and
  Stein}]{cormen2022algorithms}
Cormen TH, Leiserson CE, Rivest RL, et~al (2022) Introduction to algorithms.
  MIT press

\bibitem[{Even and Itai(1971)}]{even1971queues}
Even S, Itai A (1971) Queues, stacks and graphs. In: Theory of machines and
  computations. Elsevier, p 71--86

\bibitem[{Gagn\'{e} et~al.(2008)Gagn\'{e}, Gravel, Morin, and
  Price}]{GagneGravel}
Gagn\'{e} C, Gravel M, Morin S, et~al (2008) Impact of the pheromone trail on
  the performance of {ACO} algorithms for solving the car-sequencing problem.
  Journal of the Operational Research Society 59:1077--1090

\bibitem[{Golumbic(2004)}]{golumbic2004algorithmic}
Golumbic MC (2004) Algorithmic graph theory and perfect graphs. Elsevier

\bibitem[{Gopalan et~al.(2007)Gopalan, Jayram, Krauthgamer, and
  Kumar}]{Gopalan}
Gopalan P, Jayram T, Krauthgamer R, et~al (2007) Estimating the sortedness of a
  data stream. In: SODA

\bibitem[{Karrenbauer et~al.(2025)Karrenbauer, Kuhn, Mehlhorn, and
  Rinaldi}]{karrenbauer2025optimizingcarresequencingmixedmodel}
Karrenbauer A, Kuhn B, Mehlhorn K, et~al (2025) Optimizing car resequencing on
  mixed-model assembly lines: Algorithm development and deployment.
  \urlprefix\url{https://arxiv.org/abs/2507.17422},
  {\href{https://arxiv.org/abs/2507.17422}{{arXiv:2507.17422}}}

\bibitem[{Knuth(1968)}]{knuth68}
Knuth DE (1968) The Art of Computer Programming, Volume {I:} Fundamental
  Algorithms. Addison-Wesley

\bibitem[{Mallows(1963)}]{mallows1963patience}
Mallows CL (1963) Patience sorting. SIAM review 5(4):375

\bibitem[{Morin et~al.(2009)Morin, Gagn\'{e}, and Gravel}]{MorinGagne}
Morin S, Gagn\'{e} C, Gravel M (2009) Ant colony optimization with a
  specialized pheromone trail for the car-sequencing problem. European Journal
  of Operational Research 197:1185--1191

\bibitem[{Pratt(1973)}]{expressqueue1973}
Pratt VR (1973) Computing permutations with double-ended queues, parallel
  stacks and parallel queues. In: Proceedings of the Fifth Annual ACM Symposium
  on Theory of Computing. Association for Computing Machinery, New York, NY,
  USA, STOC '73, p 268–277, \doi{10.1145/800125.804058},
  \urlprefix\url{https://doi.org/10.1145/800125.804058}

\bibitem[{Schensted(1961)}]{schensted1961longest}
Schensted C (1961) Longest increasing and decreasing subsequences. Canadian
  Journal of Mathematics 13:179--191

\bibitem[{Tarjan(1972)}]{tarjan1972sorting}
Tarjan R (1972) Sorting using networks of queues and stacks. Journal of the ACM
  (JACM) 19(2):341--346

\bibitem[{Twelsiek(2021)}]{Twelsiek}
Twelsiek AK (2021) Production planning of mixed-model assembly lines at the
  {Ford} {Saarlouis} plant. Master's thesis, Universit\"at des Saarlandes

\end{thebibliography}

\end{document}